\PassOptionsToPackage{pagebackref,draft=false}{hyperref}
\PassOptionsToPackage{dvipsnames}{xcolor}
\documentclass[11pt]{article}
\usepackage[utf8]{inputenc}
\usepackage[english]{babel}
\usepackage[verbose=true]{microtype}
\usepackage{geometry}
\usepackage{amsfonts}
\usepackage{amssymb}
\usepackage{amsmath}
\usepackage{graphicx} 
\usepackage{amsthm}
\usepackage{authblk}
\usepackage{bbding}
\usepackage{bm}
\usepackage{color,colortbl} 
\usepackage{colonequals}	% Needed for the coloniff command
\usepackage{enumitem}
\setenumerate{itemsep=2pt,topsep=2pt}
\setitemize{itemsep=0pt,topsep=4pt}
\setlist[enumerate]{label=(\roman*)}  
\setlist[enumerate,2]{label=(\alph*)}  
\usepackage{makecell} % Tabular column heads and multilined cells
\usepackage{mathtools}
\usepackage{nameref}
\usepackage{quiver}
\usepackage[textsize=footnotesize]{todonotes}
\usepackage{tikz}
\usepackage{tikz-cd}
\usepackage[all]{xy}
\usepackage{dsfont}
\usepackage{xparse}
\usepackage{caption}
\usepackage{doi}
\usepackage[numbers, sort&compress]{natbib}
\usepackage[pagebackref,draft=false]{hyperref}
\hypersetup{breaklinks,
	colorlinks=true,
	linkcolor=blue(pigment), 
	citecolor=green(pigment), 
	filecolor=black, 
	urlcolor=blue(pigment)}

\newlength{\myparskip}
\newcounter{todocounter}
\usepackage[capitalise,noabbrev]{cleveref}

\newtheorem{theorem}{Theorem}[section]

\newtheorem{lemma}[theorem]{Lemma}
\newtheorem{proposition}[theorem]{Proposition}

\newtheorem{definition}[theorem]{Definition}
\theoremstyle{definition}
\theoremstyle{remark}

\newtheorem{remark}[theorem]{Remark}
\newtheoremstyle{question}{\myparskip}{\myparskip}{\color{BrickRed}\normalfont}{0pt}{\bfseries}{.}{5pt plus 1pt minus 1pt}{}
\theoremstyle{question}

\newtheoremstyle{answer}{\myparskip}{\myparskip}{\color{PineGreen}\normalfont}{0pt}{\bfseries}{.}{5pt plus 1pt minus 1pt}{\thmname{#1}\thmnote{ \bfseries #3}}
\theoremstyle{answer}

\definecolor{ao}{rgb}{0.0, 0.0, 1.0}
\definecolor{blue(pigment)}{rgb}{0.2, 0.2, 0.6}
\definecolor{green(pigment)}{rgb}{0.1, 0.5, 0.1}
\definecolor{ceruleanblue}{rgb}{0.16, 0.32, 0.75}
\definecolor{amber}{rgb}{1.0, 0.75, 0.0}
\definecolor{ao(english)}{rgb}{0.0, 0.5, 0.0}
\definecolor{lava}{rgb}{0.81, 0.06, 0.13}

\let\originalleft\left
\let\originalright\right
\renewcommand{\left}{\mathopen{}\mathclose\bgroup\originalleft}
\renewcommand{\right}{\aftergroup\egroup\originalright}

\newcommand{\coloniff}{\;\ratio\Longleftrightarrow\;}

\newcommand{\I}{\mathbb{I}}

\newcommand{\V}{\mathsf{V}}
\newcommand{\Z}{\mathsf{Z}}
\newcommand{\Q}{\mathsf{Q}}
\renewcommand{\S}{\mathsf{S}}
\newcommand{\A}{\mathsf{A}}
\newcommand{\D}{\mathsf{D}}
\newcommand{\C}{\mathsf{C}}
\renewcommand{\P}{\mathsf{P}}
\newcommand{\F}{\mathsf{F}}
\newcommand{\W}{\mathsf{W}}
\newcommand{\U}{\mathsf{U}}

\newcommand{\id}{\mathds{1}}

\newcommand{\im}{\mathrm{im}}

\newcommand{\ph}{\mathord{\rule[-0.05em]{0.6em}{0.05em}}}				% Argument placeholder
\ExplSyntaxOn
\NewDocumentCommand{\colvec}{m}{ 
	\begin{pmatrix}
		\seq_set_from_clist:Nn \l_tmpa_seq { #1 }
		\seq_use:Nn \l_tmpa_seq { \\ }
\end{pmatrix}}
\ExplSyntaxOff
\DeclarePairedDelimiter{\abs}{\lvert}{\rvert}
\DeclarePairedDelimiter{\norm}{\lVert}{\rVert}
\DeclarePairedDelimiterXPP{\pnorm}[2]{}{\lVert}{\rVert}{_{#1}}{#2}
\makeatletter
\let\oldabs\abs
\def\abs{\@ifstar{\oldabs}{\oldabs*}}
\let\oldnorm\norm
\def\norm{\@ifstar{\oldnorm}{\oldnorm*}}
\let\oldpnorm\pnorm
\def\pnorm{\@ifstar{\oldpnorm}{\oldpnorm*}}
\makeatother

\DeclareMathOperator{\Span}{span}

\providecommand{\given}{}			% Just to make sure the \given command exists.
\newcommand{\SetSymbol}[1][]{%
\nonscript\;\,#1\vert
\allowbreak
\nonscript\;\,
\mathopen{}
}
\DeclarePairedDelimiterX{\Set}[1]{\{}{\}}{%
\renewcommand{\given}{\SetSymbol[\delimsize]}
#1
}
\makeatletter
\let\oldSet\Set
\def\Set{\@ifstar{\oldSet}{\oldSet*}}
\makeatother

\newsavebox{\numbox}%
\newsavebox{\slashbox}%
\newsavebox{\denbox}%
\newlength{\slashlength}%
\newlength{\faktorscale}%
\DeclareDocumentCommand{\newfaktor}{m O{0.35} m O{-0.35}}{% \newfaktor{#1}[#2]{#3}[#4] -> #1/#3
	\savebox{\numbox}{\ensuremath{#1}}% Store numerator
	\savebox{\slashbox}{\ensuremath{\diagup}}% Store slash /
	\savebox{\denbox}{\ensuremath{#3}}% Store denominator
	\setlength{\faktorscale}{0.5\ht\numbox+0.5\ht\denbox}%
	\setlength{\slashlength}{2pt+0.8\faktorscale+#2\faktorscale-#4\faktorscale}%
	\raisebox{#2\ht\slashbox}{\usebox{\numbox}}% Numerator
	\mkern-2mu%
	\rotatebox{-30}{\rule[#4\ht\denbox]{0.4pt}{\slashlength}}% tilted rule as a slash
	\mkern9mu%
	\hspace{-0.44\slashlength}%
	\raisebox{#4\ht\denbox}{\usebox{\denbox}}% Denominator
}
\DeclareDocumentCommand{\linefaktor}{m O{0.08} m O{-0.08}}{% \newfaktor{#1}[#2]{#3}[#4] -> #1/#3
	\savebox{\numbox}{\ensuremath{#1}}% Store numerator
	\savebox{\slashbox}{\ensuremath{\diagup}}% Store slash /
	\savebox{\denbox}{\ensuremath{#3}}% Store denominator
	\setlength{\faktorscale}{0.5\ht\numbox+0.5\ht\denbox}%
	\setlength{\slashlength}{0.2\faktorscale+0.8\baselineskip}%
	\raisebox{#2\ht\slashbox}{\usebox{\numbox}}% Numerator
	\mkern-1mu%
	\raisebox{-0.8pt}{%
		\rotatebox{-25}{\rule[#4\ht\denbox]{0.4pt}{\slashlength}} % tilted rule as a slash
	}%
	\mkern-1mu%
	\hspace{-0.25\slashlength}%
	\raisebox{#4\ht\denbox}{\usebox{\denbox}}% Denominator
}

\begin{document}

\title{Epistemic Horizons From Deterministic Laws: Lessons From a Nomic Toy Theory}
\author[1]{Johannes Fankhauser\thanks{\href{mailto:johannes.j.fankhauser@gmail.com}{\it johannes.j.fankhauser@gmail.com}}}
\author[1]{Tom\'a\v{s} Gonda\thanks{\href{mailto:tomas.gonda@uibk.ac.at}{\it tomas.gonda@uibk.ac.at}}}
\author[1]{Gemma De les Coves\thanks{\href{mailto:gemmadelescoves@gmail.com}{\it gemmadelescoves@gmail.com}}}

\affil[1]{Institute for Theoretical Physics, University of Innsbruck, Austria}
\date{\today}
\renewcommand\Affilfont{\itshape\small}

\maketitle

\begin{abstract}
	Quantum theory has an epistemic horizon, i.e.\ exact values cannot be assigned simultaneously to incompatible physical quantities.		
	As shown by Spekkens' toy theory, positing an epistemic horizon akin to Heisenberg's uncertainty principle in a classical mechanical setting also leads to a plethora of quantum phenomena.		
	We introduce a deterministic theory\,---\,nomic toy theory\,---\,in which information gathering agents are explicitly modelled as physical systems. 
	Our main result shows the presence of an epistemic horizon for such agents.
	They can only simultaneously learn the values of observables whose Poisson bracket vanishes.
	Therefore, nomic toy theory has incompatible measurements and the complete state of a physical system cannot be known.
	The best description of a system by an agent is via an epistemic state of Spekkens' toy theory.	
	Our result reconciles us to measurement uncertainty as an aspect  of the inseparability of subjects and objects. 
	Significantly, the claims follow even though nomic toy theory is essentially classical. 
	This work invites further investigations of epistemic horizons, such as the one of (full) quantum theory. 

\end{abstract}

%\tableofcontents

\section{Introduction}
\label{sec:intro}

	Whenever an agent cannot obtain a complete account of a physical phenomenon, we shall speak of an \emph{epistemic horizon}. 
	A standard example is the Heisenberg uncertainty principle.  
	A qualitative consequence is that, given two incompatible measurements\footnotemark{} of a quantum system, an agent can only be certain about the outcome of at most one of the measurements. 
	\footnotetext{Two measurements of a quantum system described by a Hermitian operator are incompatible if the operators do not commute.}% 
There is a multitude of ways in which uncertainty about a physical system, and thus an epistemic horizon, can emerge.

	\paragraph{Sources of Epistemic Horizons.} One potential source of uncertainty arises in \textit{chaotic} systems, which exhibit high sensitivity to initial conditions. 
	Unpredictability of such systems follows due to the unavoidable inaccuracy of any specification of boundary conditions (see, for instance, \cite{Batterman1993}).

	Learning about a non-chaotic system may still be intractable because of technological limits to measurement precision. 
	Moreover, its behaviour may be unpredictable as a result of its astronomical computational complexity. 
	In both cases, the lack of knowledge an agent has about the system is connected to \textit{practical} considerations contingent on technological advances.

	\textit{Logical} `paradoxes' present another source of epistemic horizons. 
	Self-referential reasoning has been employed to establish links with undecidability, uncomputability, and randomness \cite{Svozil-acausality, szangolis-epistemic-horizon, Chiara1977-DALLSR}.
	For example, the work of Bendersky et al.\ suggests that quantum randomness must be uncomputable \cite{Bendersky-uncomputable-randomness}.
	A similar conclusion was drawn in \cite{delSanto-Gisin-randomness}, based on the idea of finite representability.

	In the context of the theory of general relativity, it has been claimed that there is an upper bound on \emph{information density}. 
	See, for instance, Bekenstein's result expressing the maximum amount of information in a bounded system \cite{Bekenstein-bounded-systems}. 
	Thus, an epistemic horizon can arise from the nature of spacetime itself for agents of bounded size.

	There are also more exotic possibilities. 
	In Everettian Quantum Mechanics and Many Worlds interpretations of quantum theory, all possible outcomes of a given measurement actually happen and are experienced independently in parallel worlds. 
	Nevertheless, our single-world experience carries a \textit{self-locating} uncertainty, which leads to uncertainty about the outcome that can be described probabilistically \cite{many-worlds-barrett-et-al}.

	In a \textit{causally indeterministic} world there is a fundamental epistemic horizon. 
	This means that events need not be pre-determined by preceding conditions together with the laws of nature \cite{sep-determinism-causal}.

	Yet another source of uncertainty is the nature of \textit{dynamical} laws. 
	For instance, in an extreme scenario of a physical theory with two types of systems without coupling, a system of one type cannot learn about the behaviour of systems of the other type when learning is mediated by interactions.
	A remote yet far-reaching example is the part of the Universe we will never interact with, which includes all systems beyond the horizon from which no information can reach us.

	Still, even in the presence of non-trivial interactions, learning faces limitations. 
	In this work we study an epistemic horizon in the context of a specific physical theory introduced below as nomic toy theory. 
	In particular, we prove that in this theory, one physical system can only obtain constrained information about another. 
	Similar to quantum theory, measurements in nomic toy theory can exhibit incompatibility. 
	Their outcomes cannot be known simultaneously by agents modelled as systems within the theory. 
	The nature of interactions of nomic toy theory thus impacts the information gathering activities of agents and entails fundamental limits to what can be known about the world.
	
	\paragraph{Dynamical Epistemic Horizons.} 
	In classical mechanics the values of the positions and momenta of all particles at a certain time, together with the physical laws, are purported to fully determine their entire future (and past) values. 
	Moreover, so the story goes, the values at a given time can be precisely measured. 
	A principal articulation of such causal determinism is the omniscient intellect of Laplace's demon \cite{Laplace1814, sep-determinism-causal}.
	
	In Newtonian physics one often ignores an explicit account of measurement interactions and that they necessarily disturb the system being measured \cite[Chapter 3]{Barad-meeting-the-universe-halfway}.	
	This is traditionally justified by stipulating that the disturbance is determinable and thus can be accounted for.
	Adjusting one's measurement record based on known disturbance\,---\,if indeed possible\,---\,allows an agent to acquire arbitrary information about a system. 
	Particularly in the context of quantum theory, measurements are said to introduce disturbance. Heisenberg's uncertainty relations were interpreted by himself as originating from an inevitable measurement disturbance: 
	Whatever pre-determines the outcome of a measurement of a particle is inadvertently disturbed by its interaction with the apparatus \cite{Heisenberg-uncertainty-principles-1925}.\footnote{However, this early account of Heisenberg's uncertainty is but one possible interpretation. 
	The properties that `cause' an individual measurement result need not exist in a quantum world. 
	In particular, it is unclear whether a single particle can be said to possess  properties of position and momentum prior to measurement (see, for example, \cite[Section 6.1]{fankhauser2022a}). Thus, one cannot straightforwardly argue that such properties (since they do not exist) would be disturbed in a measurement.}
	
	Based on complementarity, i.e.\ the existence of mutually exclusive experimental arrangements, Niels Bohr argued that the measurement disturbance in quantum theory cannot be accounted for. According to him, this is due to discontinuous quantum jumps \cite{Bohr_1937}. Thus, the discrete nature of measurement interactions spoils determinism and predictability.\footnotemark{}
	\footnotetext{Later, Heisenberg in part conceded to Bohr's views and acknowledged complementarity as the source of uncertainty (cf.\ \cite{wheeler-zurek-1983} on Heisenberg's postscript to his uncertainty article).}
	
	One perspective on our work is that it provides an account of uncertainty in Spekkens' toy theory (which reproduces stabiliser states in quantum \mbox{theory \cite{pusey2012stabilizer,catani2017spekkens}}) in terms of dynamical measurement disturbance. 
	In a nutshell, there is a classical theory\,---\,the ontological model of Spekkens' toy theory\,---\,whose deterministic laws entail an epistemic horizon. In this sense, Heisenberg's original interpretation of uncertainty can be said to apply in the case of stabiliser quantum theory.
	
	\paragraph{Spekkens' Toy Theory.} In 2004 Robert Spekkens conceived of a toy theory based on the so-called knowledge-balance principle: 
	``If one has maximal knowledge, then for every system, at every time, the amount of knowledge one possesses about the ontic state of the system at that time must equal the amount of knowledge one lacks'' \cite{Spekkens-2007-knowledge-balance} (cf.\ similar in-principle restrictions on the detectable amount of information by Zeilinger et al.\ \cite{Zeilinger-information-principle}). 
	The idea was to construct a theory in which (at least some) quantum states can be viewed as epistemic as opposed to ontic.
	That is, they would represent states of incomplete knowledge about a physical system instead of different states of physical reality. 
	The theory admits a deterministic non-contextual ontological model, whose kinematics is given by phase spaces of classical particles and its dynamics preserves the phase space structure. 
	Epistemic states of the toy theory arise from the ontic states via an epistemic restriction called classical complementarity: 
	\textit{Two linear observables on the phase space can be jointly known only if their Poisson bracket vanishes}. 
	The toy theory qualitatively reproduces a large part of the operational predictions of quantum theory \cite[Table 2]{spekkens2016quasi}. 
	For instance, it can recover the complete behaviour of states and measurements in the stabiliser subtheory of quantum theory, whose states are eigenstates of products of Pauli operators. 
	With respect to the epistemic restriction of Spekkens' theory we ask the following question: 
	\textit{Can uncertainty in a physical theory arise without imposing an a priori restriction on the acquisition of knowledge?}

	We give an affirmative answer.
	Namely, inspired by \cite{hausmann2023toys}, we provide a deterministic physical theory\,---\,nomic toy theory\,---\,and show that agents are limited in the amount of information they can gather. 
	The limitation derives from the dynamics of nomic toy theory and a definition of information gathering agents modelled within the theory. 
	Furthermore, the epistemic horizon we \emph{derive} precisely matches the \emph{postulated} epistemic restriction of Spekkens' toy theory.
	This is interesting since Spekkens' toy theory includes no formal account of how agents acquire knowledge and what is the source of the limitation. 
	To our knowledge, our work constitutes the first account of an \emph{a posteriori} epistemic horizon arising from dynamical laws. 

	\paragraph{Paper Overview.}
	We proceed as follows. 
	First, in \cref{sec:nomic_toy_theory}, we define nomic toy theory, its ontic state space, the notion of a toy system, the characterisation of agents, as well as the dynamics and the notion of a measurement. 
	In \cref{sec:epi_bound} we present the main result, which represents a fundamental epistemic horizon in nomic toy theory.
	There, we also relate our work to Spekkens' toy theory, which is shown to arise as the epistemic counterpart of our nomic toy theory (Section \ref{sec:epi_counterpart}). 
	We furthermore comment on the possibility of self-measurement in \cref{sec:self-measurement}. 
	The findings are summarised in \cref{sec:Conclusion}, where we also comment on the relationship to quantum and classical uncertainty more generally, and provide an outlook on related issues such as multi-agent scenarios and the participatory nature of the observer.  
	\cref{sec:toy_meas_examples} contains the details of a position and momentum measurement in nomic toy theory. 
	In \cref{sec:Spek_toy_theory} we provide supplementary material on Spekkens' toy theory, including several new proofs.
	For additional details on this toy theory, closely related to our nomic toy theory, we refer the reader to \cite{spekkens2016quasi,hausmann2021consolidating}. 

\section{Nomic Toy Theory}
\label{sec:nomic_toy_theory}

	To formulate our result on an epistemic horizon emerging from deterministic physical laws, we introduce nomic toy theory in which the subject-object relationship can be studied.\footnotemark{}
	\footnotetext{The use of the word nomic is motivated by the theory's emphasis on law-like interactions between an agent and another physical system.}
	The key feature of nomic toy theory is that it explicitly models the agent performing the measurement as a physical system in the theory. 
	Given the ontic state space (a classical phase space), deterministic dynamics (via symplectic maps), as well as a definition of the agent, the theory contains restrictions on what can be known about physical systems.

	We first introduce the state space and dynamics of toy systems (\cref{sec:toy_objects}) and elaborate on their properties in \cref{sec:variables}, to then define toy subjects within the theory (\cref{sec:toy_subjects}). 
	In \cref{sec:measurements} we define measurements between subjects and objects as a physical interaction.
	Finally, \cref{sec:measurable_variables} discusses what kind of information can be learned by a toy subject about a toy object via such interactions.
	An arbitrary learnable property is provided by the notion of a fixed variable (\cref{def:fix_var}).
	However, as we show in the crucial \cref{lem:fixed_from_measured}, the same information is carried by the smaller set of measurable variables (\cref{def:meas_var}).
	
	\subsection{Toy Systems}
	\label{sec:toy_objects}
	
		The formalism of physical states in nomic toy theory closely follows that of ontic states in Spekkens' toy theory (cf.\ \cref{sec:Spek_toy_systems}, \cite{hausmann2021consolidating} and \cite[Appendix A]{hausmann2023toys}).
		We begin with a description of the \textit{kinematics} of nomic toy theory and the definition of a physical system.
		\begin{definition}
			\label{def:toy_system}
			A physical system $V$ in nomic toy theory (a \textbf{toy system}) is specified by a symplectic vector space $\V$.
		\end{definition}
		We can also think of it as the phase space of a classical particle.
		Namely, $\V$ is a $2n$-dimensional $\mathcal{F}$-vector space\footnotemark{}
		\footnotetext{For a continuous toy system, $\mathcal{F}$ is $\mathbb{R}$, while for a discrete $d$-level system, it is $\mathbb{Z}_d$, in which case it is a field only if $d$ is a prime.
		For degrees of freedom with other finite cardinalities, one can instead consider $\V$ to be a $\mathbb{Z}_d$-module.} with an orthonormal basis $\{q_1,\ldots, q_n, p_1, \ldots,  p_n\}$. 
		It is furthermore equipped with a symplectic form $\omega \colon \V \times \V \to \mathcal{F}$ given by 
		\begin{equation}\label{eq:symplectic_matrix}
			\Omega = \begin{pmatrix}
					0 & \id_n \\
					- \id_n & 0
				\end{pmatrix}
		\end{equation}
		in matrix form in the above basis, where $\id_n$ is the $n \times n$ identity matrix. 
		In particular, we have
		\begin{equation} \label{eq:omega}
			\omega(x, y) = x^T \Omega y = \langle x, \Omega y \rangle ,  
		\end{equation}
		where vectors are represented as columns, $x^T$ is the transpose of $x$, and $\langle \: , \: \rangle$ is the canonical inner product. 
	
		A physical state of the toy system (an \textbf{ontic state}) is then specified by an element of $\V$.
	
		The choice of \textit{dynamics} of the theory is inspired by the Hamiltonian formulation of classical mechanics. 
		In particular, its time evolution via Hamiltonian flow is always a symplectomorphism{\,---\,}a map between symplectic manifolds that preserves the symplectic structure. 
		For the manifolds considered here, i.e.\ symplectic affine spaces, 
		there are two basic types of such transformations.
		One can be represented by a linear map $\V \to \W$ which preserves the symplectic form. 
		The other corresponds to an affine map $\V \to \V$ that translates each state by a chosen vector in $\V$.
		These are exactly the allowed transformations of ontic states in Spekkens' toy theory.
		The choice of dynamics of nomic toy theory is thus compatible with the epistemic restriction of Spekkens' toy theory (see \cref{lem:sympl_isotr,prop:maps_preserve_supp} for a proof). 
		
		For a symplectic vector space $\V$, the symplectic maps $\V \to \V$ form the symplectic group, whose matrix representation is
		\begin{equation}\label{eq:sympl_group}
			\mathrm{Sp}(\V) \coloneqq \Set*[\big]{ M \in \mathrm{GL}(\V)  \given  M^T \Omega M = \Omega },
		\end{equation}
		where $\mathrm{GL}(\V)$ is the set of the invertible linear maps of type $\V \to \V$.
				
		We thus define the group of \textbf{reversible transformations} of a given toy system in nomic toy theory to be the affine symplectic group:
		Its elements are pairs $(t,v)$ of a symplectic map $t \in \mathrm{Sp}(\V)$ and a vector $v \in \V$, which compose via 
		\begin{equation}
			(s,u) \circ (t,v) = \bigl(s \circ t, u + s(v) \bigr).
		\end{equation}
		As we can see, the dynamical evolution of ontic states in nomic toy theory is deterministic.
		That is, a given reversible transformation $(t,v)$ acts uniquely on ontic states via ${x \mapsto t(x) + v}$.

	\subsection{Properties of Toy Systems}
	\label{sec:variables}
		
		To facilitate our formal derivation of the epistemic horizon in nomic toy theory, we discuss several properties of toy systems in this section.
		Our main theorem (\cref{thm:meas_var_is_obs}) later establishes which of these properties can be acquired by a toy subject through a measurement interaction (see \cref{sec:measurements}).
		In particular, there are properties that cannot be learned in this way and thus lie beyond the epistemic horizon.
			
		Our notion of a \textit{variable} is intended to model an arbitrary property of a toy system (at a particular point in time\footnote{Note that the notion of time is implicit but of no particular relevance for the results. It only matters that a transformation connects a pre-measurement state to a post-measurement state.}). 
		On the other hand, a \textit{Poisson variable} is a special property which, as we prove later in \cref{thm:meas_var_is_obs}, is measurable by toy subjects.
		 \cref{tab:properties} provides an overview of the different kinds of properties of toy systems.
		\begin{definition}\label{def:variable}
			Let $V$ be a toy system with symplectic vector space $\V$. A function $Z \colon \V \to \Z$ is termed a \textbf{variable} of $V$, where $\Z$ is the set of values of the variable.
			A variable is termed \textbf{Poisson} if $\Z$ is an $\mathcal{F}$-vector space and $Z$ is a linear map that satisfies 
			\begin{equation}\label{eq:compatible_functionals}
				Z \Omega Z^T=0,
			\end{equation}
			where $\Omega$ is the matrix representation of the symplectic form. 
		\end{definition}

		Every variable induces a partition
		\begin{equation}\label{eq:partition}
			\Set{ Z^{-1}(x)  \given  x \in \Z }
		\end{equation}
		of the set $\V$ of ontic states.
		Variables that induce the same partition are considered to be equivalent.
		Note that Poisson variables are valued in a vector space, whose dimension tells us about the potential number of independent scalar properties it can describe.
		An important special case is when the dimension is $1$, in which case we speak of a \textbf{functional} $\V \to \mathcal{F}$.
		Such a linear map automatically satisfies \cref{eq:compatible_functionals}.

		For any basis $\{z_i\}_{i = 1}^{\mathrm{dim}(\Z)}$ of a vector space $\Z$, we can think of an arbitrary linear map $Z \colon \V \to \Z$ as a set of functionals $\{Z_i\}$, where $Z_i$ is given by $z_i^T Z$ in matrix form.	
		In this representation, \cref{eq:compatible_functionals} says that every pair of these functionals must have vanishing Poisson bracket, i.e.\ 
		\begin{equation}
			\omega(Z_i^T ,Z_j^T) = Z_i \Omega Z_j^T=0
		\end{equation}	holds for all $i$ and all $j$.
		Therefore, Poisson variables precisely correspond to properties which, in Spekkens' toy theory, are assumed to be knowable about the toy system. 

		While this epistemic horizon is traditionally postulated in Spekkens' toy theory, we derive it in nomic toy theory.
		
		\begin{remark}
			To see the connection to epistemic states of Spekkens' toy theory (\cref{sec:Spek_toy_theory}), note that the set of vectors $\{Z_i^T\}$ spans an isotropic subspace of $\V$.
			Together with a value of $Z$, it thus specifies an epistemic state.
			The support of this epistemic state is an element of the partition from \eqref{eq:partition}.
		\end{remark}
			
		\begin{table}[tb!]\centering
			\begin{tabular}{c|c|c|c} 
	%			\hline
				\thead{Property}
				& \thead{Type}
				& \thead{Values}
				& \thead{Extra conditions}
			\\ \hline 
				\makecell{Variable $Z$}
				& \makecell{Function $\V \to \Z$}
				& \makecell{Set $\Z$} 
				& \makecell{-- --}
			\\ %\hline 
				\makecell{Poisson variable $Z$}
				& \makecell{Linear map $\V \to \Z$}
				& \makecell{Vector space $\Z$} 
				& \makecell{$Z \Omega Z^T=0$}
			\\ %\hline 
				\makecell{Functional $Z$}
				& \makecell{Linear map $\V \to \mathcal{F}$}
				& \makecell{Scalar $\mathcal{F}$} 
				& \makecell{-- --}
			\\ %\hline 
			\end{tabular}
			\caption{Summary of the three different types of properties of a toy system $V$.
			Note that every functional is a Poisson variable.}
			\label{tab:properties}
		\end{table}
		
		A canonical example of a Poisson variable is the projection of $\V$ onto the $n$-dimensional subspace spanned by the $q_i$ basis vectors.
		It satisfies \cref{eq:compatible_functionals} because we have ${\omega(q_i, q_j) = 0}$ for all $i$ and all $j$.
		In other words, the symplectic form vanishes on this subspace.
		The highest dimension of a subspace with this property is $n$.
		The following standard concept generalises such a \emph{maximal} Poisson variable.
		\begin{definition}
			\label{def:Lagrangian}
			An $n$-dimensional subspace $\Q$ of a symplectic vector space $\V$, on which the symplectic form $\omega$ vanishes, is called a \textbf{Lagrangian} subspace.
		\end{definition}
		For any Lagrangian subspace, the associated projection $Q \colon \V \to \Q$ is a Poisson variable.
		Moreover, by Darboux's Theorem, there is a basis of its orthogonal complement $\P \coloneqq \Q^\perp$, in which the symplectic form has the canonical form of \cref{eq:symplectic_matrix} with respect to the decomposition $\V = \Q \oplus \P$.
	
	\subsection{Toy Subjects}
	\label{sec:toy_subjects}
	
		Physical theorising is often done from an omniscient point of view \textit{external} to the world. 
		That is, one introduces a theoretical domain of discourse\,---\,the physical world together with some law-like behaviour\,---\,to explain the phenomena that are directly observable through empirical data. 
		For instance, according to an omniscient being like Laplace's demon the future and past of the world is completely fixed if the laws are deterministic. 
		
		However, observations of phenomena necessarily occur \textit{within} the world. 
		Therefore, every physical theory requires in addition an epistemology that stipulates what can be known, e.g.\ about the physical world. 
		That is, intuitively, we need to specify what the empirical data can and cannot signify about the physical world. 
		
		And so it may happen that the two perspectives disagree. 
		Even if the omniscient viewpoint contains no fundamental uncertainty about all details of the world, an internal agent could be bound to epistemic limitations. 
		Whether the omniscient view is or is not conceivable, it may be unreachable for any agent as a result of the dynamical constraints of the world in which the agent operates.
		To study this tension, let us introduce the notion of agents in nomic toy theory.
		Note that we do not place any anthropocentric constraints on these, our agents are part of nature in the same way that their objects of study are. 
		Since our agents are decidedly minimal and may not fulfil elaborate requirements for agency \cite{van2023introducing,mcgregor2024formalising}, we also call them toy subjects.

		We only have two basic desiderata. First, our toy subjects are meant to be physical systems. They are toy systems of the same kind as any object to be observed and interacted with.\footnote{See also Hausmann et al.'s more operational approach to modelling the memory register of an agent as a toy bit \cite{hausmann2023toys}.}
		That is, an information gathering subject is an arbitrary toy system as introduced in \cref{sec:toy_objects}. 
		
		Secondly, a toy subject ought to include a specification of its `knowledge' variables. 
		These are manifest properties of the subject that represent the directly accessible empirical data on which the subject's knowledge supervenes. 
		The dynamics of nomic toy theory, in turn, dictates what the manifest variables of the subject can and cannot signify about the ontic properties of an object with which the subject interacts.

		\begin{definition}
			\label{def:toy_subject}
			A \textbf{toy subject} is a toy system $A$ equipped with a Lagrangian subspace $\Q$ of the symplectic vector space $\A$.
			The associated Poisson variable $Q \colon \A \to \Q$ is called the \textbf{manifest} variable of the subject.
		\end{definition}
		An example of a toy subject is a simple pointer apparatus. 
		The manifest variable would be the value on a scale or the angle of a pointer needle. 
		Inspired by such example, we call the manifest variable $Q$ of $A$ the \textbf{position} of $A$ and its complementary variable $P$ the \textbf{momentum} of $A$.

		Even though we do not, in general, assume what type of degrees of freedom the manifest ones are, we label them as `positions' for the sake of simplicity.
		
		The crucial property of the way we conceptualise a toy subject is that its knowledge supervenes on its manifest variable.
		Importantly, it is a variable associated to a \emph{Lagrangian} subspace of its own ontic state space.
		Therefore, the toy subject does not have direct access to the value of its own momentum variable $P$. 
		This has implications for the feasibility of measurements that the toy subject can implement.
		In particular, given a specific value $q$ of the position variable $Q$, the toy subject \emph{can} perform a measurement of another toy system conditionally, i.e.\ so that its own position prior to the measurement has value $q$. 
		On the other hand, we \emph{cannot} grant it the power to fix its own momentum value before the measurement interaction since there is no a priori way for the toy subject to know its own momentum.
		We discuss measurements in more detail in \cref{sec:measurements}.

		One may be tempted to view the restriction on a toy subject's access to its own ontic state as a kind of epistemic horizon (on self-knowledge rather than on knowledge of the world).
		However, this is not fully justified.
		Even if an agent has no direct access to some of its own degrees of freedom, it could still learn about them indirectly.
		Whether this is possible or not depends on the dynamical laws of the world in which the agent operates.
		We discuss toy subjects measuring their own momentum in the context of nomic toy theory in \cref{sec:self-measurement}.
		
		Nevertheless, that the knowledge of a toy subject supervenes on a Poisson variable (\cref{def:variable}) rather than its ontic state is a key ingredient in our derivation of the epistemic horizon in \cref{sec:main_theorem}.
		Other agents, such as ones with direct access to their own ontic state, would be able to break the epistemic horizon of nomic toy theory.

	\subsection{Measurement Interactions}
	\label{sec:measurements}
	
		Let us now turn to the discussion of how a toy subject $A$ may learn about a toy system $S$ by virtue of interacting with it.
		To distinguish $S$ from $A$, we call such $S$ the \textbf{toy object}.
		
		We model the acquisition of knowledge as a process in nomic toy theory (\cref{fig:subject-object}), which transforms the joint system of $S$ and $A$ denoted by $S \oplus A$.
		The joint ontic state space is given by the direct sum $\S \oplus \A$, which carries a canonical symplectic structure induced by those of $\S$ and $\A$.
		For more details on joint systems as well as joint and marginal states in Spekkens' toy theory, see \cref{sec:nomic_multiple_systems}.
		
		We also assume that the toy subject $A$ is in a `ready state' prior to the process, i.e.\ its position variable $Q$ has a definite value. 
		Since the value of $Q$ is already assumed to be directly accessible to $A$ (see \cref{sec:toy_subjects}), this presents no additional assumption.
		
		\begin{figure}[t]
			\centering
			\includegraphics[width=0.9\linewidth]{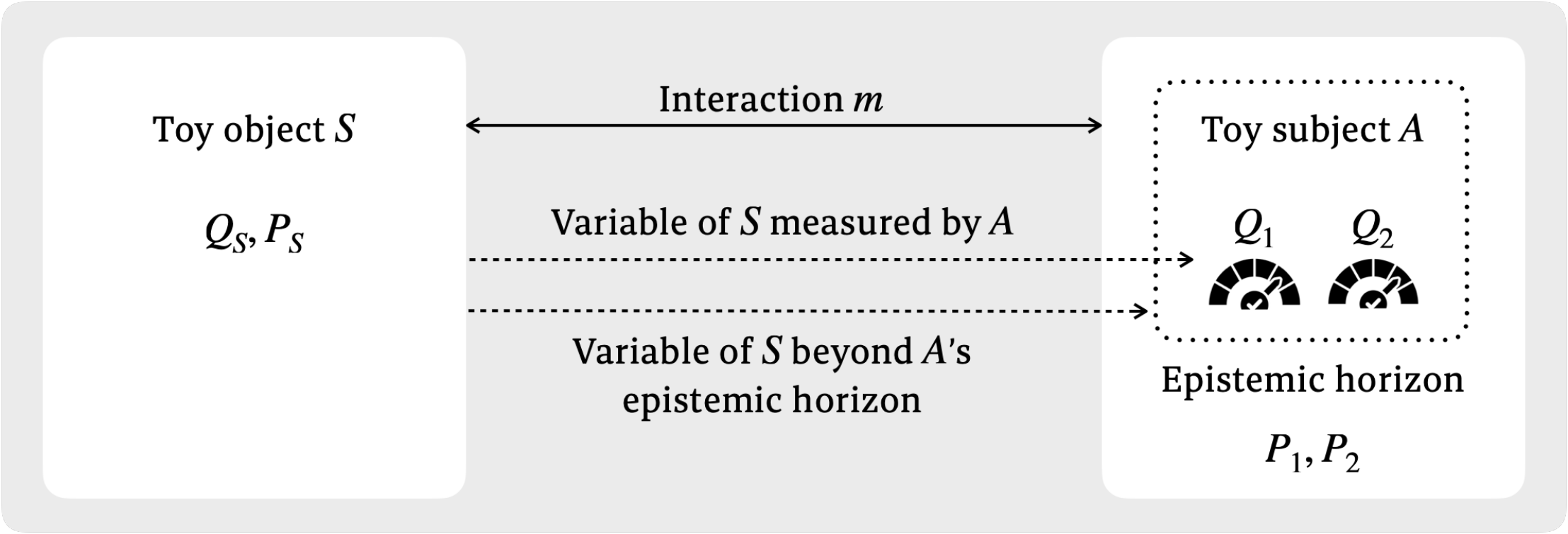}
			\caption{A toy subject $A$ with two manifest variables $Q_1$ and $Q_2$ 
			gathers information about a toy object $S$ via a measurement interaction $m$. 
			Due to the physical nature of the measurement process, the subject can only acquire information about compatible variables. 
			That is, its \textit{internal} perspective (indicated in the white box) has an epistemic horizon, by which the toy subject can learn some properties but not others. 
			In contrast, the \textit{external} omniscient perspective (gray box) features the \textit{complete} joint description of both the toy object and subject. 
			See also \cref{fig:polarised-light}.}
			\label{fig:subject-object}
		\end{figure}

		\begin{definition}\label{def:measurement}
			Given a toy system $S$ and a toy subject $A$ (with manifest variable $Q$), a \textbf{measurement} of $S$ by $A$ is a pair of a ready state, specified by a value of $Q$, and a reversible transformation $m \colon \S \oplus \A \to \S \oplus \A$ of nomic toy theory. 
		\end{definition}
		That is, $m$ is given by an affine symplectic map
		\begin{equation}
			x \mapsto M x + v
		\end{equation}
		where $x, v \in \S \oplus \A$ and $M$ is a symplectic matrix.

		The assumption that the measurement process is governed by reversible transformations does not pose any loss of generality if we assume that all irreversible transformations can be dilated to a reversible one with larger output (cf.\ \cref{sec:phys_transf}).
		Any information obtained by the irreversible process could then also be learned via its reversible dilation.
		
	\subsection{Measurable Properties of Toy Systems}
	\label{sec:measurable_variables}
	
		Regarding measurement interactions in nomic toy theory, we are concerned with the following question:
		Which variables $Z \colon \S \to \Z$ can be measured by the toy subject $A$ via a measurement as in \cref{def:measurement}?
		Our model of learning presumes that the toy subject $A$ only has direct access to its own manifest variable $Q$.
		That is, there should be a way to extract the value of $Z$ prior to the measurement from the value of $Q$ after the measurement.
		The following definition formalises this notion.
		\begin{definition}\label{def:fix_var}
			Given a measurement $m$ of a toy object $S$ by a toy subject $A$ and a variable $Z$ of $S$, we say that $Z$ is \textbf{fixed} by $m$ if there exists a function $f \colon \Q \to \Z$ satisfying
			\begin{equation}\label{eq:meas_var}
				Z(s) = f \circ Q \circ m (s + p)
			\end{equation}
			for all $s\in \S$ and all $p \in \P$. 
		\end{definition}
		Here, $Z(s)$ is the value of the $Z$ variable before the measurement took place, while ${Q \circ m (s + p)}$ is the subject's position after the measurement. Note that $\circ$ denotes, as usual, the composition of functions. 
	
		The fact that \cref{eq:meas_var} is required to hold for every $p$ expresses the assumption that the subject cannot use any direct information about its own initial momentum to learn about $Z$.

		Note that the initial value of $Q$, which has a definite value because the subject enters the interaction in a ready state, is hidden in the choice of $f$.
		Specifically, let $q_0$ be the initial poistion of the toy subject $A$.
		Given a function $f'$ satisfying
		\begin{equation}
			Z(s) = f' \circ Q \circ m (s + q_0 + p)
		\end{equation}
		for all $s$ and all $p$, one can define a new function
		\begin{equation}
			f (q) \coloneqq f' \bigl( q + Q \circ m(q_0) \bigr),
		\end{equation}
		which, by linearity of $m$ and $Q$, satisfies \cref{eq:meas_var}.
		Thus, there is no loss of generality in setting $q_0 = 0$ in \cref{def:fix_var}.
		However, the fact that the suitable $f$ depends non-trivially on $q_0$ means that assuming the toy subject to enter the interaction in a ready state is necessary.
		
		Let us decompose the measurement interaction's matrix form into blocks with respect to $\V = \S \oplus \Q \oplus  \P $ via
		\begin{equation}\label{eq:block_measurement}
			M = \begin{pmatrix}
				M_{\S \S} & M_{\S \Q} & M_{\S  \P } \\
				M_{\Q \S} & M_{\Q \Q} & M_{\Q  \P } \\
				M_{\P  \S} & M_{\P  \Q} & M_{\P   \P }
			\end{pmatrix},
		\end{equation}
		where, for example, $M_{\Q \P}$ is the block that acts as the linear map $\P \to \Q$.
		\begin{definition}\label{def:fixed_outcomes}
			Given the notation from \cref{eq:block_measurement}, the subspace $\im \left( M_{\Q\P} \right)$ of $\Q$ provides the \textbf{contingent} manifest variable given by the orthogonal projection onto this subspace and denoted by $C \colon \Q \to \C$.
			Its orthogonal complement in $\Q$ specifies the \textbf{free} manifest variable denoted by $F \colon \Q \to \F$.
		\end{definition}
Here, $\im \left( M_{\Q\P} \right)$ denotes the image of the map $M_{\Q\P}$. 
		The value of the contingent manifest variable after the transformation $m$ depends on the initial momentum of $A$, which motivates its name.
		On the other hand, the value of the free manifest variable after the measurement $m$ is independent of the initial momentum of $A$.
		
		Thus, by definition we have $\Q = \F \oplus \C$. 
		Furthermore, if we write the symplectic matrix $M$ of the transformation in a block form with respect to the decomposition 
		\begin{equation}\label{eq:fixed_random_decomposition}
			\S \oplus \A = \S \oplus \F \oplus \C \oplus \P ,
		\end{equation}
		then the block $M_{\F \P}$ vanishes by definition.
	 
		Among all the variables fixed by a given measurement, there is an essentially unique most discerning (i.e.\ most informative) one, as we show in \cref{lem:fixed_from_measured} below.
		It is the variable $s \mapsto F \circ m (s)$, which is a linear map $\S \to \F$ that is given by $M_{\F \S}$ in matrix form.
		\begin{definition}\label{def:meas_var}
			The variable \textbf{measured} by a measurement $m$ is the linear map $M_{\F \S} \colon \S \to \F$, where $F$ is the free manifest variable. 
			A variable $Z \colon \S \to \Z$ is called \textbf{measurable} if it is measured by some transformation in nomic toy theory.
		\end{definition}
		The next proposition shows that any variable fixed by a measurement can be extracted from the variable measured by it.
		Therefore, considering variables that are fixed by some measurement does not give the agent any more information about the system than merely restricting attention to variables of the form $M_{\F \S}$.
		This result justifies our identification of the set of \emph{measurable} variables as representing all properties of a toy object that a toy subject can acquire through a measurement interaction.
		\begin{proposition}\label{lem:fixed_from_measured}
			If $Z$ is a variable fixed by a measurement $m$, then there is a function $f \colon \F \to \Z$ such that for each $s \in \S$ we have
			\begin{equation}\label{eq:fixed_from_measured_1}
				Z(s) = f \left( M_{\F \S} \, s \right).
			\end{equation}
		\end{proposition}
		\begin{proof}
			Without loss of generality, we can assume that $m$ is a linear map, so that $m(x) = Mx$ for any vector $x \in \S \oplus \A$.
			This is because affine shifts do not affect whether a variable is fixed by a measurement.

			The fact that $Z$ is fixed by $m$ means that there is a function $f \colon \F \to \Z$ satisfying \cref{eq:meas_var}.
			Using the notation from \cref{eq:block_measurement} and that $M_{\F \P}$ vanishes by the definition of the free manifest variable $F$, we thus have 

			\begin{equation}\label{eq:fixed_from_measured_2}
				Z(s) = f \bigl( M_{\F\S} \, s + M_{\C\S} \, s + M_{\C\P} \, p \bigr)
			\end{equation}
			for all $s \in \S$ and all $p \in \P$.
			
			Since $M_{\C\P}$ is surjective by the definition of the contingent manifest variable $C$, there is a $p_s \in \P$ that satisfies
			\begin{equation}
				 M_{\C\P} \, p_s = - M_{\C\S} \, s 
			\end{equation}
			for a given value of $s$.
			Choosing $p$ to be $p_s$ in \cref{eq:fixed_from_measured_2} thus completes the proof.
		\end{proof}

\section{Epistemic Horizons from Deterministic Laws}
\label{sec:epi_bound}

	We are now ready to present our main result (\cref{thm:meas_var_is_obs}), which derives a limitation on the toy subject's abilities to learn about toy objects. 
	Specifically, we show that a variable is measurable (\cref{def:meas_var}) if and only if it is a Poisson variable (\cref{def:variable}) in nomic toy theory (\cref{sec:main_theorem}). 
	In \cref{sec:self-measurement}, we comment on why our agents know nothing about the object prior to learning and how this assumption can be justified with measurement disturbance.
	We also discuss a model of a toy subject measuring its own momentum and show that it does not break the epistemic horizon\,---\,unlike an agent that would have direct access to its own ontic state.
	Since Poisson variables in nomic toy theory are exactly those that can be known in Spekkens' toy theory, we conclude in \cref{sec:epi_counterpart} that Spekkens' toy theory is the epistemic counterpart of nomic toy theory.

	\subsection{Constraints on Information Acquisition}
	\label{sec:main_theorem}
	
	Recall that Poisson variables can be thought of as a collection of functionals with mutually vanishing Poisson brackets.
	Since the Poisson bracket of generic functionals does not vanish, this implies that not all properties of a toy system can be known simultaneously by an agent in the theory.
	
	With all the definitions introduced in \cref{sec:nomic_toy_theory}, we can now state our main theorem. 
	\begin{theorem}\label{thm:meas_var_is_obs}
		A variable is measurable in nomic toy theory if and only if it is a Poisson variable. 
	\end{theorem}
	Moreover, by \cref{lem:fixed_from_measured}, the only variables fixed by some measurement in nomic toy theory are those that can be written as a function of some Poisson variable. 
	We illustrate this phenomenon in \cref{fig:polarised-light}.

\begin{figure}[t]
			\centering
			\includegraphics[width=0.85\linewidth]{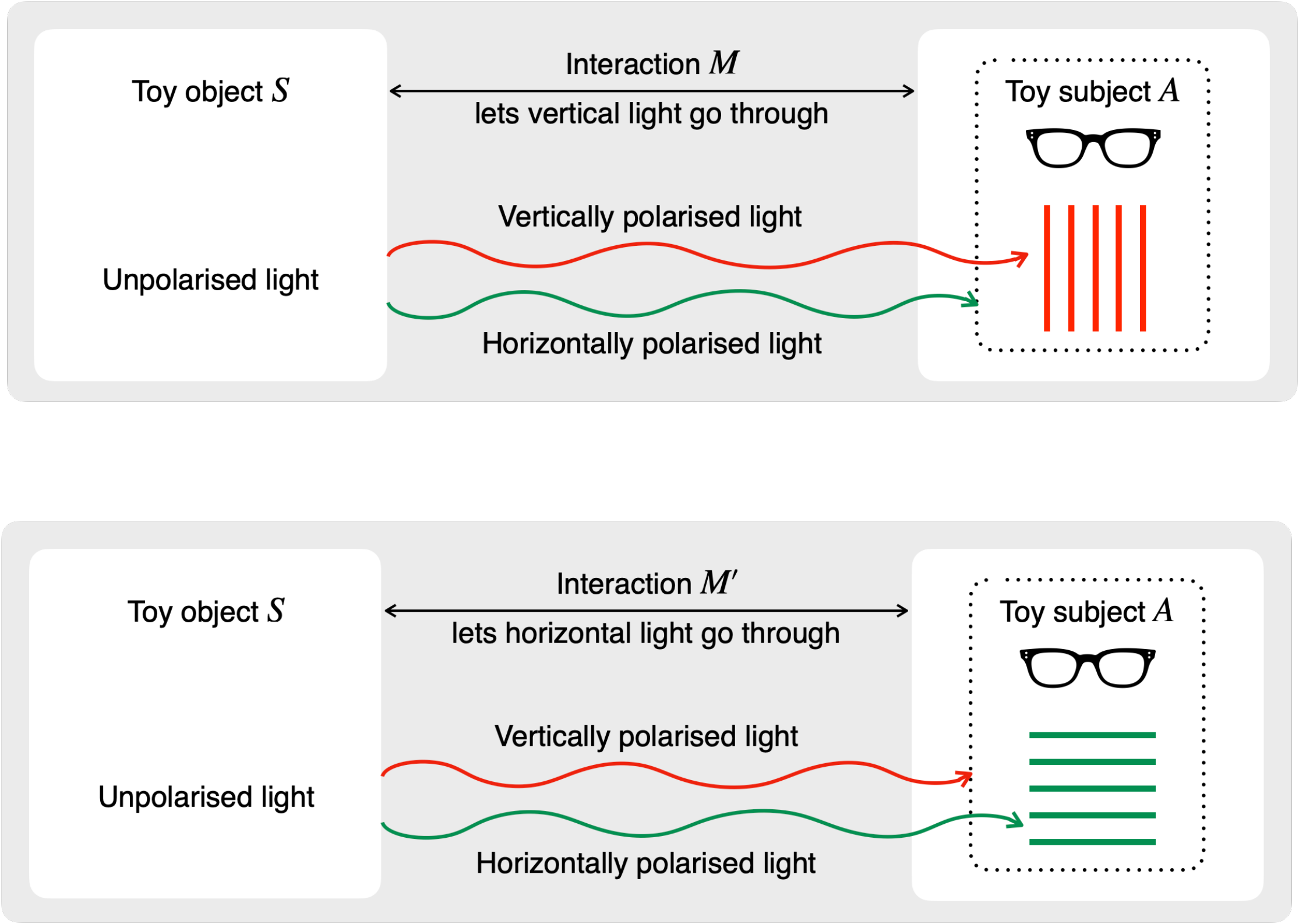}
			\caption{The epistemic horizon `experienced' by the toy subject is akin to the situation in which the subject would wear glasses that only let vertically polarised light or horizontally polarised light go through. 
			The positioning of the glasses determines whether interaction $m$ or $m'$ takes place. 
			The properties of the toy object are analogous to unpolarised light, which can be decomposed into vertical and horizontal components. 
			The toy subject can choose the orientation of glasses, but cannot observe the toy object without the glasses.}
			\label{fig:polarised-light}
		\end{figure}

	\begin{proof}
	We split the proof into two parts.
	\paragraph{Part I: Poisson variables are measurable.}
	In the first part, given any collection of \textit{compatible} components of a Poisson variable, we construct a transformation that implements their joint measurement.
	That is, we consider an arbitrary Poisson variable $Z \colon \S \to \Z$ of the toy system $S$ (see \cref{def:variable}). 
	Recall that $Z$ can be decomposed as a set of components (functionals) as $\{Z_i\}_{i = 1}^{\mathrm{dim}(\Z)}$.
	By definition, it satisfies the compatibility equation
	\begin{equation}\label{eq:comp_obs}
		Z \Omega_{\S} Z^T = 0,
	\end{equation}
	which can be interpreted as saying that its components have mutually vanishing Poisson brackets, i.e.\ they satisfy $\omega(Z_i^T, Z_j^T)=0$ for all $i$ and $j$.
	
	We then specify the phase space of the toy subject to be $\A = \Q \oplus \P$, where $\Q$ is defined to be $\Z${\,---\,}the vector space of possible values of $Z$.
	Here, we demand $\Q$ to be a Lagrangian subspace of $\A$, which thus uniquely fixes $\P$ and the symplectic structure on $\A$.
	The manifest variable of the toy subject $A$ is chosen to be the projection map $\A \to \Q$.
	
	We now construct a transformation $m \colon \S \oplus \A \to \S \oplus \A$ that measures the Poisson variable $Z$.
	It is the linear transformation given as a matrix by
	\begin{equation}\label{eq:meas_comp_obs}
		M =  \begin{pmatrix}
			\id & 0 & \Omega_{\S} Z^T \\
			Z & \id & 0 \\
			0 & 0 & \id
		\end{pmatrix}
	\end{equation}
	in the block form relative to the decomposition $\S \oplus \A = \S \oplus \Q \oplus \P$.
	A specific example of the above matrix for the case of a position measurement can be found below.
	
	Note that we have $M_{\Q\P} = 0$, which implies $\F = \Q$, and $M_{\Q\S} = Z$. 
	Thus, by \cref{def:meas_var}, the transformation $m$ measures $Z$ if it is indeed a valid transformation in nomic toy theory.
	To show that it is, we have to prove that it is a symplectic matrix, i.e.\ that ${M^T \Omega M = \Omega}$ holds.
	The left-hand side of this equation gives
	\begin{equation}
		\begin{pmatrix}
			\Omega_{\S} & 0 & 0 \\
			0 & 0 & \id  \\
			0 & - \id & Z \Omega_{\S} Z^T
		\end{pmatrix},
	\end{equation}
	which is indeed equal to $\Omega$, provided that $Z$ is a Poisson variable satisfying \cref{eq:comp_obs}.
	
	\paragraph{Part II: Measurable variables are Poisson variables.}
	In the second part of the proof, we show that no other variables can be measured by valid transformations in nomic toy theory.
	
	Our task is to show that if $Z$ is measurable, then it must be Poisson, which means proving 
	\begin{equation}
		Z \Omega_{\S} Z^T=0,
	\end{equation}
	since the fact that $Z$ is a linear map follows from the definition of measurable variables.
	
	Consider now a measurement $m$ where the toy subject $A$ is given by the symplectic vector space $\Q \oplus \P$ where $Q$ is its manifest variable.
	Moreover, the linear part of $m$ is denoted by $M$ with blocks denoted with respect to the decomposition from \eqref{eq:fixed_random_decomposition}.
	The fact that $m$ is a transformation in nomic toy theory means that $M$ is a symplectic matrix.
	Moreover, the transpose of every symplectic matrix is also symplectic, i.e.\ we have $M \Omega M^T = \Omega$.
	Extracting the $\F\F$ block out of this set of $16$ equations, we find
	\begin{equation}
		M_{\F\S} \Omega_{\S} M_{\F\S}^T - M_{\F\P} \left( M_{\F\F} + M_{\F\C} \right)^T + \left( M_{\F\F} + M_{\F\C} \right) M_{\F\P}^T = 0.
	\end{equation}
	Since $M_{\F\P}$ is the zero matrix by \cref{def:fixed_outcomes}, this implies 
	\begin{equation}\label{eq:sympl_7}
		M_{\F\S} \Omega_{\S} M_{\F\S}^T = 0,
	\end{equation}
	which is what we wanted to show, concluding the proof of \cref{thm:meas_var_is_obs}.
\end{proof}

	Note that every functional is a Poisson variable. 
	\Cref{thm:meas_var_is_obs} thus implies that every functional is measurable. 
	Furthermore, by the construction in the first part of the proof, a 2-dimensional subject suffices to measure it.
	
	\paragraph{An example of a position measurement.}
	Let us illustrate the construction of the measurement of a generic Poisson variable $Z$ with a concrete example.
	To this end, consider both the toy object $S$ and the toy subject $A$ to be 2-dimensional, i.e.\ each one comes with a single position and a single momentum degree of freedom. 
	Moreover, we choose $Z$ to be the position variable of $S$, which is a functional in this case. 
	In matrix form, $Z$ is given by 
	\begin{equation}\label{eq:pos_var}
		\begin{pmatrix}
			1 & 0
		\end{pmatrix}
	\end{equation}
	in the $\{q_\S, p_\S\}$ basis of $\S$.
	
	Before the measurement, the initial joint state of $S \oplus A$ is denoted by
	\begin{equation}
		v = \colvec{v_1, v_2, v_3, v_4}
	\end{equation}
	in the $\{q_\S, p_\S, q_\A, p_\A\}$ basis of $\S \oplus \A$.
	On the other hand, the measurement interaction from the proof of \cref{thm:meas_var_is_obs} is in general given by the matrix $M$ from \cref{eq:meas_comp_obs}.
	Substituting the position variable from \eqref{eq:pos_var} for $Z$ in this expression gives
	\begin{equation}\label{eq:pos_meas}
		M =  \begin{pmatrix}
			1 & 0 & 0 & 0 \\
			0 & 1 & 0 & -1 \\
			1 & 0 & 1 & 0 \\
			0 & 0 & 0 & 1 
		\end{pmatrix}.
	\end{equation}
	Thus, the post-measurement ontic state of the joint system is 
	\begin{equation}
		M v = \colvec{v_1, v_2 - v_4, v_1 + v_3, v_4}.
	\end{equation}
	We notice two crucial features. 
	First, if the toy subject $A$ is initially in a ready state, i.e.\ if $v_3$ has a definite value, then the manifest variable after the measurement encodes the initial position of $S$ given by $v_1$.
	This illustrates one role of our assumption that the agent's manifest variable be fixed prior to the measurement.
	
	Secondly, there is a back-reaction on the object's momentum\,---\,the conjugate variable to the measured position of $S$.
	In particular, its value after the measurement is disturbed by a value that equals the initial momentum of the toy subject $A$.
	This disturbance highlights the role of our assumption that the toy subject cannot directly know its own momentum.
	If it did, the measurement disturbance could be accounted for.
	
	Let us discuss both of these points in more detail now.

\subsection{A Couple of Caveats}
\label{sec:self-measurement}
	
	Our claim that the learning of an agent in nomic toy theory is limited by an epistemic horizon hinges on the following caveat. 
		
	\paragraph{The Relevance of (No) a Priori Knowledge.}
	We assume that, prior to any measurement, the agent possesses no knowledge about the state of the toy object $S$.
	Indeed, imagine that, on the contrary, the following is true:
	The agent $A$ is composed of two subsystems, i.e.\ we have $\A = \A_1 \oplus \A_2$ where each $A_i$ is a toy subject with an associated manifest variable $Q_i$.
	At time $t_1$ (labelling that the measurement process is yet to occur), the value of the manifest variable $Q_1$ encodes the momentum of $S$ and, importantly, the agent $A$ knows that this is the case.
	Furthermore, $A_2$ is in a ready state (see \cref{def:measurement}).
	
	Then, applying the transformation $m_q \colon \S \oplus \A_2 \to \S \oplus \A_2$ given by \cref{eq:meas_comp_obs}, where $Z$ is the position variable of $S$, enables $A_2$ to learn the position of $S$.
	Since the state of $A_1$ is unchanged by this transformation, at time $t_2$ (labelling that the transformation has occurred) we have the following situation:
	The manifest variable $Q_1$ at time $t_2$ encodes the momentum of $S$ at time $t_1$ and the manifest variable $Q_2$ at time $t_2$ encodes the position of $S$ at time $t_1$.
	In conjunction, at time $t_2$, the agent $A$ has a complete specification of the ontic state of $S$ at time $t_1$ and thus breaks the purported epistemic horizon.
	
	However, it is natural to assume that the agent has no knowledge of the toy object's state initially.
	After all, we are interested in deriving fundamental bounds on the learning capabilities of the agent. 
	Any pre-existing knowledge should be accounted for by an explicit process that allows $A$ to obtain information.
	Hence, we can circumvent the above caveat on grounds that the acquisition of information invariably involves an interaction with the physical world. 
	This justifies our assumption that toy subjects possess no a priori knowledge of the state of the toy object. 

	\paragraph{Measurement Disturbance.}
	But how can we be certain that the knowledge of the object's momentum by $A_1$ could not be accounted for by an explicit process in nomic toy theory?
	Abstractly, this follows from \cref{thm:meas_var_is_obs}.
	
	More concretely, consider a measurement $m_p \colon \S \oplus \A_1 \to \S \oplus \A_1$ that is applied before $m_q$ and given by \cref{eq:meas_comp_obs}, where $Z$ is now the momentum of $S$ (see \cref{fig:pos_mom_meas}).
	That is, $m_p$ encodes the momentum of $S$ at time $t_0$ (labelling that the measurement $m_p$ is yet to occur) into the manifest variable $Q_1$ at time $t_1$. 
	One can check that the momentum of $S$ is unaffected by $m_p$ and thus the manifest variable $Q_1$ at time $t_1$ also coincides with the momentum of $S$ at time $t_0$.
	See \cref{sec:toy_meas_examples} for the explicit computations.
	
		\begin{figure}[t]\centering
			\includegraphics[width=1\columnwidth]{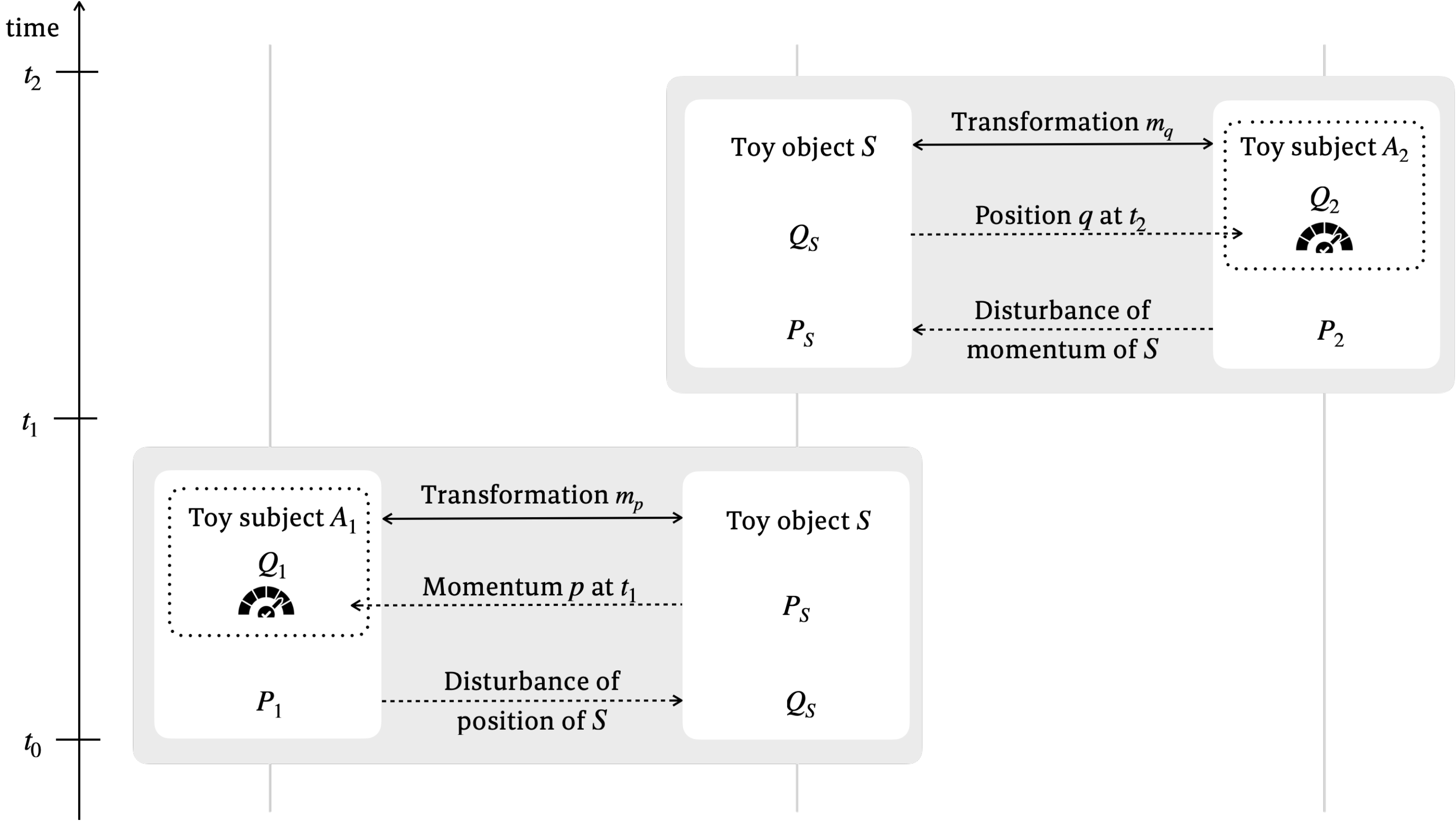}
			\caption{
				At time $t_0$, both subjects are in a ready state with $Q_1 = Q_2 = 0$.
				Then, the toy subject $A_1$ measures the toy object $S$, so that the initial value of the object's momentum gets encoded into the manifest variable $Q_1$.
				At the same time, the toy object's position gets disturbed by an amount that depends on the initial momentum $P_1$ of the toy subject $A_1$.
				Subsequently, the other toy subject $A_2$ interacts with $S$, so that the object's position at time $t_1$ gets encoded into the manifest variable $Q_2$.
				During this interaction, the object's momentum is disturbed by an amount that depends on the initial momentum of $A_2$.}
			\label{fig:pos_mom_meas}
		\end{figure}

	The issue is that $m_p$ disturbs the position of the toy object $S$ by a shift that depends on the (unknown) momentum of $A_1$ at time $t_0$.
	As a result, the subsequent transformation $m_q$ encodes the position of $S$ at time $t_1$, rather than the one at time $t_0$, into the manifest variable $Q_2$.
	The composed process $m_q \circ m_p$ is therefore no counterexample to \cref{thm:meas_var_is_obs}.
	Nevertheless, our discussion shows that the toy subject $A$ can, at time $t_2$, have perfect knowledge of the ontic state of the toy object $S$ at time $t_1$!
	This is in line with the fact that also the ontic state in Spekkens' toy theory can be perfectly known given both pre- and post-selection~\cite{hausmann2023toys}. It is worth mentioning that the same is true in quantum mechanics.
	
	Nevertheless, neither the ontic state of $S$ at time $t_0$ nor the one at time $t_2$ is completely known to $A$. That is, it is still true that the subject $A$ cannot at any time encode a previously unknown ontic state of the object $S$ that it possesses at that very same time.
	The transformation $m_p$ applied between times $t_0$ and $t_1$ disturbs the position of $S$, while the transformation $m_q$ applied between times $t_1$ and $t_2$ disturbs the momentum of $S$.
	The results we derive in this paper rule out the possibility that after a measurement (time $t_2$ above), $A$ would know the ontic state of $S$ before the measurement (time $t_0$ above), which is what we call learning.
	The fact that $A$ may have sufficient information to determine the value of incompatible variables at an intermediate time during the measurement is not ruled out by the epistemic horizon we derive.

	\paragraph{Self-Measurement of Toy Subjects.} 
	The reason why $A$ cannot access the initial position of $S$ in the above example is its disturbance by the initial momentum of $A_1$, which is unknown.
	But could a toy subject measure its own momentum and with this information correct for the disturbance?
	
	As it turns out, such a self-measurement is implicitly accounted for by \cref{thm:meas_var_is_obs}.
	Our main result therefore gives an abstract argument why measuring one's own momentum cannot break the epistemic horizon.
	
	This stems from measurement disturbance again.
	Specifically, imagine a further toy subject $A_3$ that measures the momentum of $A_1$ before time $t_0$.
	This measurement disturbs the position of $A_1$ by a shift that depends on the (unknown) initial momentum of $A_3$.
	As a result, the momentum of $S$ is no longer fixed by $m_p$.
	Self-measurement thus does not resolve the issue, the source of the uncertainty has been merely shifted from the initial momentum of $A_1$ to that of $A_3$.
	One could imagine introducing further pointers to measure initial momenta, but this inevitable leads to an infinite regress that does not stabilise to a reliable knowledge of the relevant parameters.

	Our analysis of nomic toy theory implies that an epistemic horizon exists also in classical mechanics, given the assumption that agents modelled as mechanical systems can only directly access their own manifest variable (\cref{def:toy_subject}). 
	In contrast, traditional accounts claim that in classical mechanics arbitrary measurement precision can be achieved and that both position and momentum can be recorded simultaneously (see, for instance, \cite{measurement-disturbance-zanghi}). 
	However, a closer look at these arguments reveals that this holds only under the assumption that the initial momentum of the measurement apparatus is known\,---\,in line with our discussion above.
	
\subsection{Spekkens' Toy Theory as the Epistemic Counterpart of the Nomic Toy Theory}
\label{sec:epi_counterpart}

In this section we briefly discuss the connection between variables measurable in nomic toy theory and epistemic states in Spekkens' toy theory \cite{spekkens2016quasi}. 

While agents are not explicitly modelled as physical systems in Spekkens' toy theory, its epistemic restriction is introduced to specify what a hypothetical agent could learn about a physical system. 
	
Ontic states and the associated reversible transformations in Spekkens' toy theory match those of nomic toy theory.
	While the latter posits no a priori notion of epistemic states, these are explicitly specified in Spekkens' toy theory (see our description in \cref{sec:Spek_toy_theory} for more details).
	Specifically, each epistemic state corresponds to the value of a variable that can be known according to the epistemic restriction in Spekkens' toy theory (and vice versa).
	Among these `knowable' variables, the scalar-valued ones are called \textit{quadrature functionals}.
	A generic one, an affine map $f \colon \V \to \mathcal{F}$, can be written as
	\begin{equation}
		f \coloneqq a_1 q_1^T + b_1 p_1^T + ... + a_n q_n^T + b_n p_n^T + c,
	\end{equation} 
	where $\{q_1, \ldots, q_n, p_1, \ldots,  p_n \}$ is the chosen orthonormal basis of the phase space $\V$ and $a_j$, $b_j$, $c$ are all scalars in the field $\mathcal{F}$.
	As far as the resulting epistemic state is concerned, we can assume $c=0$ without loss of generality (cf.\ our notion of equivalence of variables introduced in \cref{sec:variables}).

	Generic (vector-valued) linear variable can be identified as a collection of quadrature functionals.
	The epistemic restriction of Spekkens' toy model says that such a collection is jointly knowable if and only if the Poisson bracket of each pair of them vanishes, i.e.\ $\{f_1,f_2\}=0$ in Spekkens' standard notation for quadrature functionals. 
	The theory postulates that variables whose value can be known are precisely Poisson variables in nomic toy theory as introduced in \cref{def:variable}.
	Furthermore, as we show in \cref{lem:fixed_from_measured,thm:meas_var_is_obs}, Poisson variables in nomic toy theory coincide with those properties of toy systems that can be learned by a toy subject within the world.
	In this way, the epistemic restriction of Spekkens' toy theory arises from two ingredients:
	\begin{itemize}
		\item[(1)] the allowed transformations of ontic states introduced in \cref{sec:toy_objects} (which coincide for nomic and Spekkens' toy theories), and
		\item[(2)] the specification of information gathering agents and identification of their directly accessible information in the form of manifest variables (\cref{def:toy_subject}).
	\end{itemize}
	While nomic toy theory does not come with a pre-specified epistemology, the second ingredient allows us to derive an epistemic horizon for the model of toy subjects used in this article.
	Doing so, we find that the derived epistemic aspects of subjects in nomic toy theory coincide\,---\,at least as far as epistemic states are concerned\,---\,with the posited epistemic horizon in Spekkens' toy theory.

\section{Conclusions}		
\label{sec:Conclusion}

	Let us now discuss implications of our results and related questions in the foundations of physics so as to put things into a broader perspective. 
	We discuss the significance of our work for the relationship of internal and external observers, representationalism, the subject-object split and the reality of unobserved properties (\cref{sec:discussion}). 
	We also comment on the relationship of nomic toy theory to quantum theory, and a possible view of physical phenomena that supersedes the subject-object separability.
	We then conclude with an outlook on future directions of study in \cref{sec:outlook}.
	
	\subsection{Discussion}
	\label{sec:discussion}
	
		\paragraph{Internal vs.\ External Perspective.} 
		\Cref{thm:meas_var_is_obs} can be interpreted as describing a relationship of two distinct perspectives. 
		One is the omniscient view that specifies the precise ontic state of every system in toy world, akin to the meticulous vision of the entire state of the toy universe by Laplace's demon. 
		This view is by definition from `outside', i.e.\ external to the world.
		Conversely, there is an internal perspective as experienced by an embedded toy subject. 
		This view is shown to be limited relative to the omniscient one. 
		As we prove, a subject in toy theory cannot learn the precise ontic state of another toy system by interacting with it.
		The best description it can have is an epistemic state, which necessarily retains uncertainty about the precise ontic state (see \cref{sec:Spek_toy_systems} for details).

		\paragraph{Subject-Object Inseparability.} 
		The	derived epistemic horizon emphasises the participatory nature of the subject in the theory. 
		It shows that the physically allowed information gathering activities of an agent affect the knowledge it can have about an object.
		This challenges the old Cartesian divide between the subject and object. 
		That is, our approach highlights that the standard notions of measurement, representability, and epistemology are intimately bound up. 

		Moreover, the construction of a toy subject measuring itself (\cref{sec:self-measurement}) introduces the possibility of self-reference, which in turn makes the knowledge of a toy subject liable to logical paradoxes.
		It is conceivable that our results could be linked to a logical argument about the impossibility for an observer to describe itself from within the world.
		In particular, recall that the crucial \cref{def:toy_subject} of toy subjects specifies what a toy subject knows about its own ontic state as well as how its knowledge is manifested in its ontic state. Relatedly, Ismael presents an argument for the instability in an embedded agent’s ability to know the future due to self-reference \cite{ismael2023open}.
		
		So it could be argued, perhaps, that what is `real' to one subject is not `real' to another. Furthermore, does it make sense for the subject to speak of a world as being separate from itself? 
		 What would a measurement outcome signify if we take the participatory nature of the subject seriously and abandon an observation-independent reality? What is the new referent of measurement? In other words, what supersedes the Cartesian subject-object split?

		\paragraph{Epistemic Horizons and their Implications for Ontology.}

			The idea that a physical theory may operate under the premise of an observer-dependent description is not new.
			Several interpretations of quantum theory take a similar stance, such as the non-realist \cite{Fuchs-Schack-QBism, Rovelli-RQM, sep-qm-copenhagen}, pragmatist \cite{Healey-pragmatist}, or Everettian \cite{Everett-QM} approaches.

		Nomic toy theory gives an explicit account of the interdependency of subjects and objects. 
		It invites us to study whether subjects are justified to posit the existence of ontic states that are only `visible' from an omniscient perspective. 
		Even though subjects in nomic toy theory are faced with an epistemic horizon, this limitation is compatible with a deterministic and classical description. 
		Can the same be said for other kinds of subjective experiences featuring an epistemic horizon, such as the one of quantum theory?
		Are there operational theories whose predictions \emph{rule out} the possibility that their epistemic horizon stems from the dynamical laws of a classical ontic theory?

		Making these questions precise requires a careful construction of a more general framework than our investigation of nomic toy theory and its symplectic dynamics.
		With it, one may hope to classify the kinds of epistemic horizons that could arise based on the allowed subject-object interactions just like the one we derive in this paper.
		Similar efforts have been successful in the framework of ontological models (a.k.a.\ hidden variable models), in which one can formally derive the operational consequences of metaphysical assumptions such as Bell locality \cite{Bell_Aspect_2004} and non-contextuality \cite{Kochen-Specker-theorem}.
		
		Importantly, the fact that the operational consequences of both are violated by behaviours of quantum systems constrains the possible underlying physical reality.
		Answering the questions from previous paragraphs would likely constitute an analogous step in understanding quantum theory and its viable interpretations.
		
		\paragraph{Relation to Interpretations of Quantum Theory.}
		Although we do not provide answers to the questions posed in the previous paragraphs, it is worthwhile to mention that a version of an observer-dependent realism aligns with the spirit of relational and QBist approaches to quantum mechanics \cite{Rovelli-RQM, Fuchs-Schack-QBism}.
		See also Barad's agential realism \cite{Barad-meeting-the-universe-halfway} and the quantum holism of Ismael and Schaffer \cite{ismael2020quantum}. 
		
		For	instance, relational quantum mechanics purports that the notion of a subject has no metaphysical significance{\,---\,}any physical system could be one.
		Moreover, it emphasises ``the way in which one part of nature manifests itself to any other single part of nature'' \cite[p.\ 67]{rovelli2021helgoland}. 
		In this view, properties of an object are relative to another system which interacts (and thus measures) the object. This resonates with the notion of the observer-dependent epistemic state in nomic toy theory. 
		
		Relational quantum mechanics, as well as many other interpretations, effectively posit that quantum properties \textit{do not exist} prior to measurement or that there is no way to consistently describe them (see, for instance, Wheeler's participatory nature \cite[pp.\ 182--213]{wheeler-zurek-1983}). 
		This is in contrast with the ontic status of unobservable variables in Spekkens' toy theory (and thus also nomic toy theory). 
		There, we have an epistemic horizon featuring unpredictability, uncertainty, and complementarity, even though all properties of systems exist and have definite values at all times (at least from the omniscient perspective featuring the full ontic state description). From a toy subject's perspective, however, the view is very similar to one invoking participatory `realism'. The Cartesian subject-object divide can be therefore called into question even given a deterministic physical theory.

		Furthermore, recall the intuition that the epistemic horizon of nomic toy theory is connected to the uncontrollable initial momentum of toy subjects, which introduces an unpredictable disturbance of the toy object (\cref{sec:self-measurement}). 
		This implies that a toy subject measuring position after a measurement of momentum (\cref{fig:pos_mom_meas}) may find a different position value than a toy subject measuring position prior to the measurement of momentum. 
		More generally, \cref{thm:meas_var_is_obs} implies that there is no simultaneous measurement of both position and momentum{\,---\,}they are incompatible.
		In quantum theory, the incompatibility structure of observables leads to contextuality \cite{Kochen-Specker-theorem}. 
		In contrast, Spekkens' toy theory is non-contextual{\,---\,}all measurements have pre-determined outcomes.
		In this case, the incompatibility can be seen merely as an expression of measurement disturbance.

\paragraph{Limitations on Predictability.}
	We suspect that our result also implies that a toy subject cannot \textit{prepare} the toy object in a fixed ontic state. 
	Intuitively, this would follow from the fact that in order to prepare an exact ontic state, a subject would need to perform a measurement that signifies the preparation of this state. But as we have shown such a measurement process does not exist (a similar claim was proven by \cite{hausmann2023toys} in the context of Spekkens' toy theory).

	\subsection{Summary and Outlook}
	\label{sec:outlook}
		We have used nomic toy theory{\,---\,}an essentially \textit{classical} theory{\,---\,}to propose an explicit account of the source of the epistemic horizon in Spekkens' toy theory. 
		Subjects in nomic toy theory can only ever ascertain a coarse-grained description of objects in the world, namely one in terms of the epistemic states of Spekkens' toy theory. 
		We attribute the source of the fundamental uncertainty to the nature of interactions between subjects and objects.
		Specifically, the learning process governed by such an interaction is invariably connected to a disturbance of the object, which prevents the subject from learning the complete state of the object. 

		At first glance, our result may be surprising in light of the claims that Newtonian mechanics should in principle allow for arbitrarily precise measurements of the properties of a classical particle. 
		Bear in mind that Liouville's theorem in Hamiltonian mechanics implies preservation of phase space volume, but does not rule out arbitrary stretching and squeezing of a phase space volume such that conjugate variables become simultaneously sharply defined.
		However, we suspect that our result could be related to the claims of de Gosson on the relationship of symplectic geometry and quantum uncertainty principles \cite{Gosson-symplectic-camel-uncertainty}. 
		Basically, de Gosson derived an analogue of the quantum Robertson--Schrödinger inequality from the symplectic properties of the phase space alone. 
		This essentially implies that Heisenberg's uncertainty relations already hold in Hamiltonian mechanics for all pairs of conjugate position and momentum variables.

		Why does it seem that some aspects of quantum uncertainty can be explained in terms of Hamiltonian mechanics? 
		Do uncertainty relations really have such an analogue in classical physics? 
		Can the epistemic horizon in nomic toy theory be restated as a classical uncertainty relation akin to Heisenberg's uncertainty principle in quantum theory? 
		We hope that our analysis will serve as a toy example to facilitate explorations of those pertinent questions.

		For instance, one may study the role of hidden variables in quantum theory. 
		Our result derives the consequences of positing a specific classical ontology for the learning capabilities of internal agents. 
		Not all underlying ontic models may lead to the same information gathering capabilities of agents. 
		Thus, the empirically observed epistemic horizon could potentially be used to rule out the ontological models that do not reproduce it. 
		More on this is found in the ``implications for ontology'' part of the Discussion (\cref{sec:discussion}).

		A related question concerns the development of ontological models motivated and evaluated from within nomic toy theory. 
		That is, one may investigate what kind of ontologies are consistent with the experience of an epistemically restricted toy subject.
		Is there a way to differentiate among them based on desiderata such as parsimony or naturalness? 
		In a nutshell, what would such a subject conclude about the ontology underlying the phenomena observed? 
		See also a potential link to problems about bootstrapping and reliabilist epistemology \cite{sep-reliabilism}.
		
		As one possibility, one could look at nomic toy theory in an Everettian setting where pointer states are not single valued. 
		Could a many-worlds ontology lead to a single-world experience of the toy subject \cite[Chapter 2]{many-worlds-barrett-et-al}? 
		It would be interesting to study the problem of Everettian probabilities in this context \cite[Chapter 3]{many-worlds-barrett-et-al}. 
		There may also exist connections to more elaborate models of agents such as those in \cite{shrapnel2023stepping}.
		
		In the future we also wish to shed light on multi-agent scenarios.
		A recent attempt to try to view quantum theory as an integration of perspectives of agents subject to the epistemic horizon of Spekkens' toy theory has been explored in \cite{braasch2022classical}.
		It is particularly interesting to look at what different subjects can communicate intersubjectively (see also related ideas in the context of Spekkens' toy model \cite{hausmann2023toys}). 
		This may perhaps allow novel insights into the intricacies of many recently studied Wigner's friend type scenarios as well as no-go claims on `observer-independent facts' \cite{Wigner1961-mind-body, Cavalcanti-local-friendliness-no-go, frauchiger2018quantum, Lawrence2023relativefactsof, Brukner-observer-independent, ormrod2022nogo}.
		See also the reviews in \cite{adlam-absoluteness,  schmid2023review, Brukner-relational-objectivity} and more general results on quantum epistemic boundaries \cite{fankhauser2023epistemic}.
		
		We also leave open the question whether the participatory nature of the agent in our toy theory entails a more parsimonious account of the physical world. 
		Could there exist a new \textit{relational} physical state of the world relative to the internal observers of the theory describing the subject and object jointly?
		Such an account would go beyond the old Cartesian split and take the inseparability of subjects and objects seriously.

\section*{Acknowledgements}

	This work was supported by the Start Prize Y1261-N of the Austrian Science Fund (FWF).
	For open access purposes, the authors have applied a CC BY public copyright license to any accepted manuscript version arising from this submission.

\bibliographystyle{abbrvnat}
\bibliography{biblio-paperalpha}

\appendix

\section{Composing Position and Momentum Measurements}
\label{sec:toy_meas_examples}
	
	Here, we give additional details on the attempted construction of a joint measurement of position and momentum from \cref{sec:self-measurement}.
	Specifically, we consider three toy systems\,---\,$S$, $A_1$, and $A_2$\,---\,each of which has one position and one momentum degree of freedom.
	Moreover, the latter two are toy subjects with their positions acting as manifest variables (see \cref{fig:pos_mom_meas}).
	
	The joint system $A_1 \oplus S \oplus A_2$ starts out at time $t_0$ in the ontic state denoted by
	\begin{equation}
		u(t_0) = \colvec{u_1, u_2}_{\A_1} \oplus \colvec{u_3, u_4}_{\S} \oplus \colvec{u_5, u_6}_{\A_2}
	\end{equation}
	in the $\{q_1, p_1, q_{\S}, p_\S, q_2, p_2\}$ basis of $\A_1 \oplus \S \oplus \A_2$.
	
	The first interaction $m_p$ is a measurement of the momentum of $S$ by the toy subject $A_1$.
	Just as at the end of \cref{sec:main_theorem}, we substitute the matrix form 
	\begin{equation}
		\begin{pmatrix}
			0 & 1
		\end{pmatrix}
	\end{equation}
	of the momentum variable into \cref{eq:meas_comp_obs} to obtain the matrix form of $m_p$:
	\begin{equation}
		M_p = 
			\begin{pmatrix}
				1 & 0 & 0 & 1 \\
				0 & 1 & 0 & 0 \\
				0 & 1 & 1 & 0 \\
				0 & 0 & 0 & 1 
			\end{pmatrix} ,
	\end{equation}
	where one ought to be careful that the subject and object are now in reverse order compared to \cref{eq:meas_comp_obs}.
	Here, we merely write its action on $\A_1 \oplus \S$.
	The action on the full joint state space is then via $M_p \oplus \id_{\A_2} $.

	At time $t_1$, i.e.\ once the interaction $m_p$ has taken place, the joint state of all three toy systems is thus
	\begin{equation}
		u(t_1) = \colvec{u_1+u_4, u_2}_{\A_1} \oplus \colvec{u_2 + u_3, u_4}_{\S} \oplus \colvec{u_5, u_6}_{\A_2}.
	\end{equation}
	As we can see, the manifest variable of $A_1$ now encodes the initial momentum of $S$, provided that $A_1$ started out in a ready state.
	Furthermore, the position of $S$ has been disturbed by the initial momentum of $A_1$.

	The second step of the composite transformation depicted in \cref{fig:pos_mom_meas} is a measurement $m_q$ of the position of $S$ by the toy subject $A_2$.		
	Its matrix form is as in \cref{eq:pos_meas}:
	\begin{equation}
		M_q =  \begin{pmatrix}
			1 & 0 & 0 & 0 \\
			0 & 1 & 0 & -1 \\
			1 & 0 & 1 & 0 \\
			0 & 0 & 0 & 1 
		\end{pmatrix}. 
	\end{equation}
	After this interaction, at time $t_2$, the full ontic state is  
	\begin{equation}
		u(t_2) =  \colvec{u_1 + u_4, u_2}_{\A_1} \oplus \colvec{u_2 + u_3 , u_4-u_6}_{\S} \oplus \colvec{u_2+u_3+u_5, u_6}_{\A_2}.
	\end{equation}	
	If we assume the ready states of the toy subjects have vanishing manifest variables, this reduces to 
	\begin{equation}
		\colvec{u_4, u_2}_{\A_1} \oplus \colvec{u_2 + u_3, u_4-u_6 }_{\S} \oplus \colvec{u_2+u_3, u_6}_{\A_2}.
	\end{equation}
	The values of the manifest variables at time $t_2$ are thus $u_4$ and $u_2+u_3$ respectively.
	The former encodes the correct momentum of $S$ at times $t_0$ and $t_1$, while the latter encodes the correct position of $S$ at times $t_1$ and $t_2$.

\section{Supplementary Material on Spekkens' Toy Theory}
\label{sec:Spek_toy_theory}
	
	As we mention throughout the text,  nomic toy theory shares both the kinematics and dynamics with Spekkens' toy theory \cite{Spekkens-2007-knowledge-balance}.
	This is not an accident.
	We are specifically interested in the latter because it features both an epistemic restriction as well as deterministic dynamics at the ontic level.
	As we discuss in \cref{sec:epi_counterpart}, our results show that the epistemic restriction of Spekkens' toy theory coincides with the epistemic horizon of nomic toy theory that we derive.
	To make this precise, we provide a description of the epistemic level of Spekkens' toy theory here including several auxilliary results.
	Our presentations closely follows that of \cite{hausmann2021consolidating}.
	For additional details on Spekkens' toy theory, see \cite{spekkens2016quasi,catani2017spekkens}.
	
	\subsection{Systems in Spekkens' Toy Theory}
	\label{sec:Spek_toy_systems}

		The ontic state space of a system $V$ is a symplectic vector space $\V$, just as we discuss in \cref{sec:toy_objects}. 
		\begin{remark}
			If the underlying field of $\V$ is that of real numbers, we obtain continuous toy systems.
			Basic finite systems are associated with an integer $d$. 
			Their ontic state space is a (symplectic) $\mathbb{Z}_d$-module, which is a vector space if $d$ is a prime power.
			Other finite systems can be obtained as composites of the basic ones (see \cref{sec:composition}).
		\end{remark}
		
		For a linear subspace $\W$ of $\V$, we can define the \textbf{symplectic complement}
		\begin{equation}
			\W^{\omega} \coloneqq \Set{ v \in \V  \given  \omega ( \W, v ) = 0}
		\end{equation}
		where 
		\begin{equation}
			 \omega ( \W, v ) = 0  \quad \coloniff \quad   \omega ( x, v ) = 0 \quad \forall \, x \in \W.
		\end{equation}
		Such a subspace $\W$ is 
		\begin{itemize}
			\item a \textbf{symplectic subspace} if $\W^{\omega} \cap \W = \{0\}$,
			\item \textbf{isotropic} if $\W \subseteq \W^{\omega}$, i.e.\ if the symplectic form vanishes on $\W$, and
			\item Lagrangian if $\W = \W^{\omega}$, i.e.\ if it is a maximal isotropic subspace (cf.\ \cref{def:Lagrangian}). 
		\end{itemize}
		
		An \textbf{epistemic state} of Spekkens' toy theory $(\U,a)$ is specified by an isotropic subspace $\U$ of $\V$ and a vector $a \in \V$.
		Via an isomorphism of $\V$ and its dual $\V^*$, the subspace $\U$ is interpreted as consisting of those functionals whose values are known.
		Alternatively, we can think of $\U$ as the set of values of the orthogonal projection $U \colon \V \to \U$.
		This is an isotropic variable if and only if $\U$ is isotropic.

		The vector $a$ is interpreted as one of the ontic states that is deemed possible by this epistemic state.

		It fixes the value of any functional $u \in \U$ to be
		\begin{equation}
			\langle u, a \rangle
		\end{equation}
		where $\langle \ph, \ph \rangle$ is the canonical inner product on $\V$.
		Thus, the set of all ontic states that are possible according to the epistemic state $(\U,a)$ is
		\begin{equation}\label{eq:support}
			\U_a \coloneqq \Set*[\big]{ v \in \V  \given  \langle u, a \rangle = \langle u, v \rangle  \;\; \forall \, u \in \U } = \U^\perp + a,
		\end{equation}
		where $\U^\perp$ is the orthogonal complement of $\U$.
		In other words, the possible ontic states must share the value of the variable $U$.
		We call $\U_a$ the \textbf{support} of the epistemic state $(\U,a)$.
		Note that it is an affine subspace of $\V$.
		We do not distinguish between epistemic states that have the same support.
		An epistemic state $(\U,a)$ is called \textbf{pure} if $\U$ is Lagrangian.
		
	The reversible transformations of Spekkens' toy theory form the affine symplectic group (\cref{sec:toy_objects}) and act on ontic states via the canonical action.
		That is, its elements are pairs $(t,v)$ of a symplectic map $t \in \mathrm{Sp}(\V)$ and a vector $v \in \V$, which compose via 
		\begin{equation}
			(s,u) \circ (t,v) = \bigl(s \circ t, u + s(v) \bigr).
		\end{equation}
		A given reversible transformation $(t,v)$ then acts on ontic states via ${x \mapsto t(x) + v}$.
		\begin{figure}[t]\centering
		\includegraphics[width=.65\columnwidth]{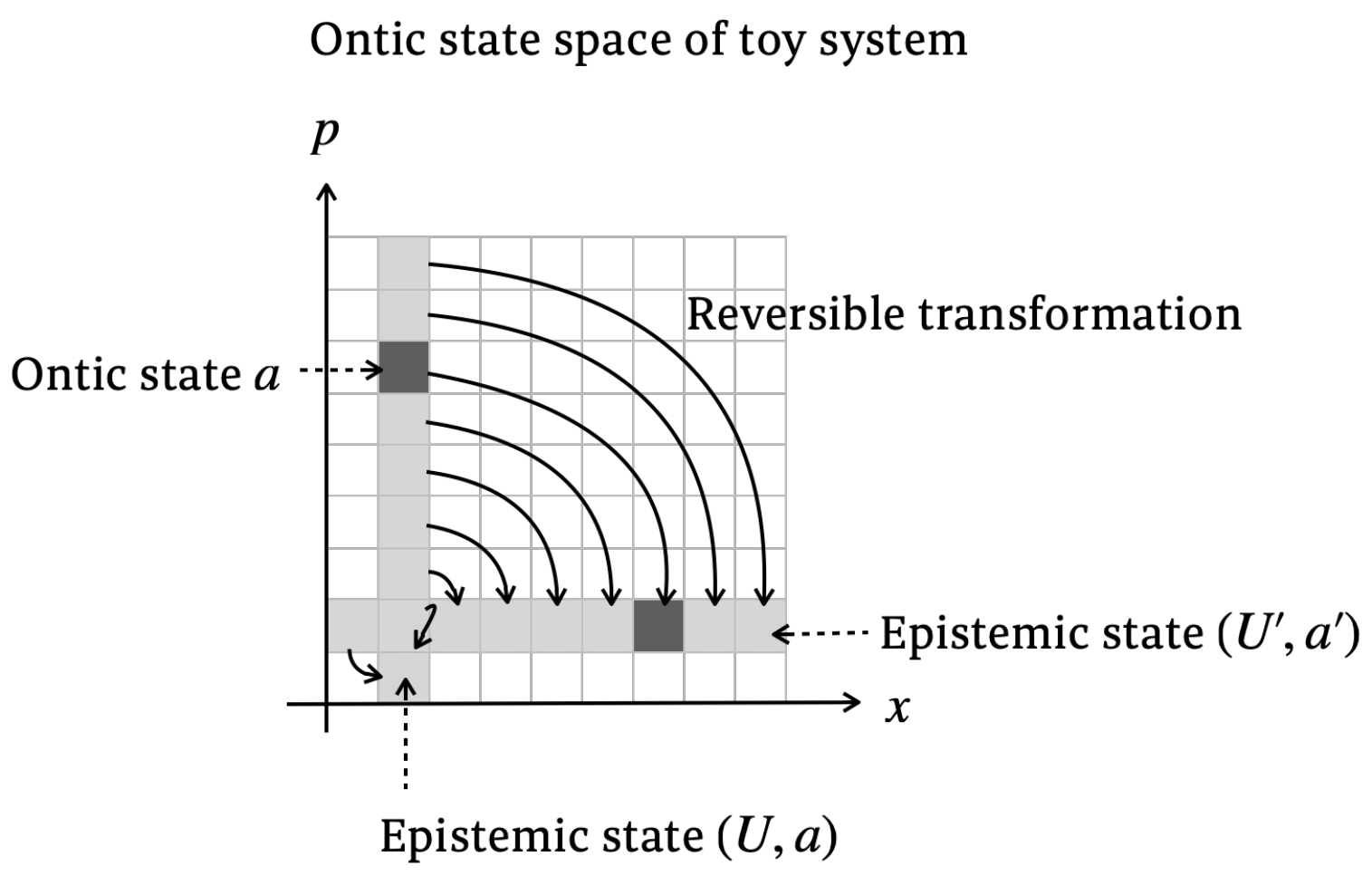}
		\caption{While the ontic state $a$ is an element in the ontic state space of the toy system, the support of an epistemic state $(\U,a)$ is a collection of such elements, namely those given by \cref{eq:support}. 
			After a reversible transformation, these are transformed to another collection $(\U',a')$ given by \cref{eq:maps_preserve_supp}, where $a'$ is the image of $a$ under the transformation.
			In the figure, $U$ is the position variable and $U'$ is the momentum variable.}
		\label{fig:ontic-state-space}
		\end{figure}
		
		The following proposition shows that epistemic states are mapped to epistemic states under affine symplectic transformations and derives Equation (A.3) \mbox{from \cite{hausmann2023toys}}.
		\begin{proposition}\label{prop:maps_preserve_supp}
			Let $(f,v)$ be an affine symplectic map on a symplectic vector space $\V$ and let $\U_a$ be the support of an epistemic state $(\U,a)$.
			The affine subspace $f(\U_a) + v$, which is the image of $\U_a$ under $(f,v)$, coincides with the support of the epistemic state
			\begin{equation}\label{eq:maps_preserve_supp}
				\left( {f^{-1}}^T(\U), f(a) + v \right).
			\end{equation}
		\end{proposition}	

		\begin{proof}
			First, let us show that \eqref{eq:maps_preserve_supp} is indeed an epistemic state.
			To this end, note that the inverse of any symplectic matrix $M \in \mathrm{Sp}(\V)$ is given by
			\begin{equation}\label{eq:symp_inv}
				M^{-1} = \Omega^T M^T \Omega.
			\end{equation}
			Therefore, ${f^{-1}}^T$ is given by 
			\begin{equation}\label{eq:inv_transp}
				{f^{-1}}^T (v) = - \Omega F \Omega v
			\end{equation}
			where $F$ is the matrix representation of $f$.
			In particular, it is also a symplectic map.
			By \cref{lem:sympl_isotr} proven below, the image of $\U$ under ${f^{-1}}^T$ is an isotropic subspace and \eqref{eq:maps_preserve_supp} is thus an epistemic state.
			
			The rest of the proof establishes that $f(\U_a) + v$ is the support of this epistemic state.
			We have
			\begin{equation}\label{eq:maps_preserve_supp_1}
				f(\U_a) + v = \Set*[\big]{ f(x) + v  \given  \langle u, a \rangle = \langle u, x \rangle \;\; \forall \, u \in \U }
			\end{equation}
			by definition.
			Let us denote $f(x) + v$ by $w$, so that we have $x = f^{-1}(w-v)$.
			Then, the right-hand side of \cref{eq:maps_preserve_supp_1} is the set of all $w \in \V$ satisfying
			\begin{equation}\label{eq:maps_preserve_supp_2}
				\langle u, a \rangle = \left\langle u, f^{-1}(w-v) \right\rangle \;\; \forall \, u \in \U.
			\end{equation}
			Since we have $f^T {f^{-1}}^T = \I$, the left-hand side of \cref{eq:maps_preserve_supp_2} is equal to either side of 
			\begin{equation}
				\left\langle f^T {f^{-1}}^T (u), a \right\rangle = \left\langle {f^{-1}}^T (u), f(a) \right\rangle,
			\end{equation}
			while the right-hand side of \cref{eq:maps_preserve_supp_2} is 
			\begin{equation}
				\left\langle {f^{-1}}^T (u), w - v \right\rangle .
			\end{equation}
			Thus we obtain
			\begin{equation}\label{eq:maps_preserve_supp_3}
				f(\U_a) + v = \Set{ w  \given  \left\langle {f^{-1}}^T (u), f(a) + v \right\rangle = \left\langle {f^{-1}}^T (u), w \right\rangle \;\; \forall \, u \in \U },
			\end{equation}
			which is the support of the epistemic state in \eqref{eq:maps_preserve_supp}.
		\end{proof}
		Therefore, reversible transformations preserve the set of epistemic states.
		In other words, if a function $f \colon \V \to \V$ maps (the support of) some epistemic state to a subset of $\V$ that is not (the support of) an epistemic state, then $f$ is not a valid reversible transformation.
		For example, this directly implies Corollary 1 (Restrictions on conditional transformations: example) in \cite{hausmann2023toys}.

	\begin{lemma}\label{lem:sympl_isotr}
		Symplectic maps preserve the set of isotropic subspaces.
		That is, if $f \colon \V \to \V$ is a symplectic map and $\W$ is an isotropic subspace of $\V$, then $f(\W)$ is also an isotropic subspace.
	\end{lemma}
	This is a standard result, we give the proof for completeness.
	\begin{proof}
		Note that a subspace $\W$ is isotropic if and only if the implication
		\begin{equation}
			v \in \W  \quad \implies \quad  \omega (x , v) = 0  \quad  \forall \, x \in \W
		\end{equation}
		holds.
		Moreover, since $f$ is bijective, we have $y \in f(\W)$ if and only if $y = f(v)$ for some $v \in \W$.
		Thus we have
		\begin{align}
			y \in f(\W)  \quad &\,\implies \quad  \omega \left( x , f^{-1}(y) \right) = 0  \quad  \forall \, x \in \W \\
			&\iff \quad  \omega \left( f^{-1}(z) , f^{-1}(y) \right) = 0  \quad  \forall \, z \in f(\W) \\
			&\iff \quad  \omega \left( z , y \right) = 0  \quad  \forall \, z \in f(\W) 
		\end{align}
		where the last equivalence holds because $f^{-1}$ is itself symplectic.
		In conclusion, $f(\W)$ is isotropic.
	\end{proof}
	Note, furthermore, that symplectic maps in $\mathrm{Sp}(V)$ act transitively on the Lagrangian Grassmanian \cite[Lemma 1.12]{calegari2022notes}.	
		
\subsection{Description of multiple systems in Spekkens' toy theory}	
\label{sec:nomic_multiple_systems}

	\subsubsection{Joint states}\label{sec:composition}

	Each ontic state of the joint (bipartite) system is given by a pair of ontic states from each of the components respectively.
	Its underlying vector space is thus the direct product of the individual ones, which is isomorphic to their direct sum.
	\begin{definition}
		Given two toy systems $(\V_1, \omega_1)$ and $(\V_2, \omega_2)$, the \textbf{joint system} describing their composite is given by $(\V_1 \oplus \V_2, \omega_1 \oplus \omega_2)$.
	\end{definition}
	Every joint ontic state $v \in \V \coloneqq \V_1 \oplus \V_2$ has a unique decomposition $v = v_1 + v_2$ for $v_i \in \V_i$.
	Moreover, there are linear projections $V_i \colon \V \to \V_i$, such that $V_i$ maps $v$ to $v_i$. 
	Similarly, for any choice of epistemic states $(\U_1,a_1)$, and $(\U_2,a_2)$ of $\V_1$ and $\V_2$ respectively, the joint state of $\V_1 \oplus \V_2$ is the epistemic state $(\U_1 \oplus \U_2, a_1 + a_2)$ with support\footnotemark{}
	\begin{equation}
		\left( \U_1 \oplus \U_2 \right)^\perp + a_1 + a_2 = \left( \U_1^\perp + a_1 \right) \oplus \left( \U_2^\perp + a_2 \right).
	\end{equation}
	\footnotetext{Note that on the right-hand side, $\U_i^\perp$ refers to the orthogonal complement of $\U_i$ \emph{within} $\V_i$, as opposed to the left-hand side, where it denotes the orthogonal complement in $\V$.}%
	These constitute the so-called \textbf{product states} of the joint system.
	
	Besides product states, there are also correlated joint states.
	As an example, consider the joint system of two toy bits with its epistemic state given by
	\begin{equation}
		\left( \Span\left\{ \colvec{1,0,-1,0}, \colvec{0,1,0,-1} \right\}, \colvec{0,0,0,0} \right). 
	\end{equation}
	It is a state for which both the positions and momenta of the two systems are perfectly correlated.
	Its support is the subset of $(\mathbb{Z}_2)^4$ given by
	\begin{equation}
		\left\{ \colvec{0,0,0,0}, \colvec{0,1,0,1}, \colvec{1,0,1,0}, \colvec{1,1,1,1} \right\} .
	\end{equation}

	\subsubsection{Reduced states}\label{sec:marginalization}
	
	\Cref{sec:composition} describes global states of multiple systems. 
	For any such global state, we can marginalize any of its subsystems to obtain the local description of the remaining subsystems.
	This notion also appears in \cref{def:measurement} of pointer-preserving measurements.
	
	\begin{definition}
		Given a possibilistic state\footnotemark{} $\rho$ of a composite system $\V_1 \oplus \V_2$, its $\bm{\V_i}$\textbf{-marginal} (also referred to as the reduced state to $\V_i$) is the image of $\rho$ under the projection $V_i$.
	\end{definition}
	\footnotetext{A possibilistic state of a toy system is a subset of its underlying vector space of ontic states.
		Key examples of possibilistic states are supports of epistemic states.}%
	
	Whenever $\rho$ is the support of an epistemic state $(\U,a)$, we can find its marginal by projecting $a$ and restricting the set of known functionals in $\U$ to the local ones.
	\begin{proposition}[Marginals of epistemic states]\label{lem:ep_st_marginal}
		Consider an epistemic state $(\U,a)$ of the composite $\V = \V_1 \oplus \V_2$.
		Then the $\V_1$-marginal of its support is the support of the epistemic state of $\V_1$ given by
		\begin{equation}\label{eq:ep_st_marginal}
			\bigl( \U \cap \V_1, V_1(a) \bigr).
		\end{equation}
	\end{proposition}
	\begin{proof}
		The $\V_1$-marginal of (the support of) $(\U,a)$ is 
		\begin{equation}
			V_1 \bigl( \U^\perp \bigr) + V_1(a),
		\end{equation}
		while the support of the epistemic state in \eqref{eq:ep_st_marginal} is 
		\begin{equation}\label{eq:ep_st_marginal_2}
			\left( \U \cap \V_1 \right)^\perp + V_1(a),
		\end{equation}
		where the orthogonal complement is within $\V_1$ here.
		The task is to show that these two affine subspaces of $\V_1$ coincide.
		Writing expression $\eqref{eq:ep_st_marginal_2}$ instead in terms of the orthogonal complement within $\V$, we thus have to show
		\begin{equation}\label{eq:ep_st_marginal_3}
			V_1 \bigl( \U^\perp \bigr) = \left( \U \cap \V_1 \right)^\perp \cap \V_1.
		\end{equation}
		
		It is an elementary fact that $\left( \U \cap \V_1 \right)^\perp = \U^\perp + \V_1^\perp$ holds, see for example \cite[Lemma B.3]{hausmann2023toys}.
		Therefore, we can rewrite the right-hand side of \cref{eq:ep_st_marginal_2} as 
		\begin{equation}
			\left( \U^\perp + \V_1^\perp \right) \cap \V_1,
		\end{equation}
		which can be further simplified as follows
		\begin{align}
			\left( \U^\perp + \V_1^\perp \right) \cap \V_1 &= \left( V_1 \bigl( \U^\perp \bigr) \oplus \V_1^\perp \right) \cap \V_1 \\
			&=  V_1 \bigl( \U^\perp \bigr),
		\end{align}
		because $V_1 \bigl( \U^\perp \bigr)$ is a subspace of $\V_1$ and $\V_1^\perp$ is orthogonal to $\V_1$.
		Thus, we get the desired equality.
	\end{proof}

\subsection{General physical transformations}\label{sec:phys_transf}

	\Cref{sec:toy_objects} introduces the reversible transformations of nomic toy theory (which are identical to those of Spekkens' toy theory).
	A generic physical transformation may also involve discarding of subsytems, and as a result become irreversible.
	\begin{definition}\label{def:phys_transf}
		A \textbf{physical transformation} between two toy systems given by symplectic vector spaces $\V$ and $\W$ respectively is an affine symplectic map $\V \to \W$. 
	\end{definition}
	\begin{proposition}\label{prop:phys_transf}
		An affine map $f \colon \V \to \W$ is a physical transformation if and only if there is a decomposition $\V \cong \W \oplus \W^\perp$, a reversible transformation $\tilde{f} \in \mathrm{Sp}(\V)$, and an $w \in \W$ satisfying 
		\begin{equation}\label{eq:phys_transf}
			f(v) = W \circ \tilde{f}(v) + w,
		\end{equation}
		where $W \colon \V \to \W$ is the symplectic, orthogonal projection of $\W \oplus \W^\perp$ onto $\W$.
	\end{proposition}
	\begin{proof}
		The ``if'' direction is immediate. 
		For the ``only if'' part, note that since the symplectic form is non-degenerate and the symplectic part of $f$ preserves it, the image of $f$ must coincide with $\W$.
		Thus, by the first isomorphism theorem, we have 
		\begin{equation}
			\W \cong \newfaktor{\V}{\mathrm{ker}(f)},
		\end{equation}
		which implies $\V = \W \oplus \W^\perp$ as symplectic vector spaces, since $\mathrm{ker}(f) = \W^\perp$ is necessarily a symplectic subspace.
		
		Now we can let $\tilde{f} \coloneqq f|_{\W} \oplus \id_{\W^\perp}$, which satisfies \cref{eq:phys_transf}.
	\end{proof}
	One of the consequences of \cref{prop:phys_transf} is that the dimension of $\V$ cannot be smaller than the dimension of $\W$.
	Another is that every physical transformation has a reversible dilation given by $v \mapsto \tilde{f}(v) + w$.

\subsection{Measurable variables in nomic toy theory are copyable}\label{sec:info_var}

	\begin{definition}\label{def:info_var}
		A variable $Z$ is an \textbf{information variable} if there is a reversible transformation $f \colon \V \oplus \V \to \V \oplus \V$ and an epistemic state $(\U, a)$ of $V$, satisfying
		\begin{equation}\label{eq:info_var}
			 Z(v) \oplus Z(v)  = (Z \oplus Z) \circ f(v + x) 
		\end{equation}
		for every ontic state $v \in \V$ and every ontic state $x$ in the support of $(\U, a)$.
	\end{definition}
	In other words, information variables carry information that can be copied.

	The following result says that a variable is copyable if and only if it is a collection of functionals whose Poisson brackets vanish.
	Together with \cref{thm:meas_var_is_obs}, it entails that a variable in nomic toy theory is measurable if and only if it is an information variable.
	\begin{proposition}\label{prop:copy}
		A variable is an information variable if and only if it is a Poisson variable.
	\end{proposition}
	\begin{proof}
		First of all, let us show that if $Z \colon \V \to \Z$ is a Poisson variable, then it is an information variable, i.e.\ that it is copyable.
		
		To this end, we denote the vector space $\bigl( \mathrm{ker}(Z) \bigr)^\perp \cong \linefaktor{\S}{\mathrm{ker}(Z)}$ by $\D$ and the isomorphism arising from the first isomorphism theorem by ${K \colon \D \to \mathrm{im}(Z)}$.
		Since $Z$ is a Poisson variable, $\D$ must be an isotropic subspace of $\V$.

	The copying of $Z$ is then achieved by the transformation ${M \colon \V \oplus \V \to \V \oplus \V}$ introduced in \cref{eq:meas_comp_obs} and using $K^{-1} Z \colon \V \to \D$ instead of its $M_{\Q \P}$ component. 
	It is symplectic for the same reason as in the proof of \cref{thm:meas_var_is_obs}, i.e.\ by virtue of $Z$ being a Poisson variable.
	Applying this $M$ to an arbitrary ready state input, written 
	\begin{equation}\label{eq:copy_input}
		\begin{pmatrix}
			v \\
			0 \\
			x
		\end{pmatrix}
	\end{equation}
	in the $\V \oplus \D \oplus \mathrm{ker}(Z)$ decomposition, gives 
	\begin{equation}\label{eq:copy_output}
		\begin{pmatrix}
			s + \Omega_\V Z^T \left( K^{-1} \right)^T x \\
			s_\D \\
			x
		\end{pmatrix},
	\end{equation}
	Here, $s_\D$ denotes $K^{-1} Z s$, which is the orthogonal projection of $s$ onto $\D$.
	Applying the variable $Z \oplus Z$ to the output state yields
	\begin{equation}\label{eq:copy_1}
		Z \left( s + \Omega_\V Z^T \left( K^{-1} \right)^T x \right ) = Z s
	\end{equation}	
	in the first instance of $\V$ and 
	\begin{equation}\label{eq:copy_2}
		Z ( s_\D + x ) = Z s
	\end{equation}
	in the second instance.
	\Cref{eq:copy_1} follows because of the property $Z \Omega_\V Z^T = 0$ satisfied by every Poisson variable.
	\Cref{eq:copy_2} is a consequence of $ Z x = 0$ (since $x$ is in the kernel of $Z$) and $Z s_\D = Z s$ (since $s_\D$ is the orthogonal projection of $s$ onto $\D$).
	As a result, we have shown that $Z$ is an information variable.
	
	Let us now prove the converse.
	Namely, we assume that $Z$ is an information variable.
	Then the copying transformation, composed with applying $Z$ to the second instance of $\V$, is a transformation that measures $Z$.
	Therefore, by \cref{thm:meas_var_is_obs}, $Z$ is a Poisson variable.
\end{proof}

\subsection{Pointer-preserving measurements}\label{sec:intrinsic_meas}
	For any value $a$ of the position $Q$ of a toy subject $A$, we can associate an epistemic state $(\Q,a)$ of Spekkens' toy theory, called a \textbf{pointer state} of $A$.
	The name comes from the fact that we think of the manifest variable of the toy subject also as a pointer of a measurement apparatus. 
	
	A specific class of measurements in nomic toy theory are those that preserve pointer states.
	As we show in \cref{rem:block_measurement} below, they have the special property that the contingent manifest variable is trivial. 
	In particular, for a pointer-preserving measurement, the post-measurement manifest variable of $A$ is independent of its initial momentum, as a result of the property $M_{\Q \P} = 0$.
	For this reason, they are transformations that measure the variable $M_{\Q \S}$ (see \cref{def:meas_var}).
	The characterisation of pointer-preserving measurements from \cref{rem:block_measurement} says that they satisfy an additional property, namely that $M_{\P \P}$ is non-degenerate. 
	This suggests that there may be other transformations besides pointer-preserving measurements that also measure the variable $M_{\Q \S}$.

\begin{definition}\label{def:intr_meas}
	A \textbf{pointer-preserving measurement} consists of an affine symplectic transformation $m \colon \S \oplus \A \to \S \oplus \A$ and a Lagrangian subspace $\Q$ of $\A$, such that for every $s \in \S$, the associated map $\A \to \A$ given by the composite (here, $A \colon \S \oplus \A \to \A$ denotes the projection map)
	\begin{equation}\label{eq:meas_pointer_map}
			\begin{tikzcd}[row sep=1ex]
			\A & {\S \oplus \A} & {\S \oplus \A} & \A \\
			a & {s + a} & {m(s + a)} & {A \circ m (s+a)}
			\arrow["{s \oplus \id}", from=1-1, to=1-2]
			\arrow["m", from=1-2, to=1-3]
			\arrow["{A}", from=1-3, to=1-4]
			\arrow[maps to, from=2-1, to=2-2]
			\arrow[maps to, from=2-2, to=2-3]
			\arrow[maps to, from=2-3, to=2-4]
		\end{tikzcd}
	\end{equation}
	maps (the support of) each pointer state $(\Q,a)$ to that of another pointer state $(\Q, a')$.

\end{definition}
\begin{proposition}[Characterisation of pointer-preserving measurements]
	Let $M$ be the matrix representation of the linear part of the transformation $m$ above.
	Then the following are equivalent:
	\begin{enumerate}
		\item $m$ and $\Q$ make up a pointer-preserving measurement.
		\item Given $q \in \A$, we have $M^T q \in \S \oplus \Q$ if and only if $q \in \Q$.
	\end{enumerate}
\end{proposition}
\begin{proof}
	Let us analyse how a pointer state $(\Q,a)$ is transformed by the respective maps in the composite \eqref{eq:meas_pointer_map}.
	After the first step, adjoining the ontic state $s$ to the possibilistic state $\P + a$ (note that we have $\P = \Q^T$ by definition of the momentum variable) leads to the possibilistic state\footnotemark{}
	\footnotetext{Note that here we use the same notation for possibilistic states as for epistemic states.
	That is, for any subspace $\U$ (which need not be isotropic), the possibilistic state associated to $(\U,a)$ is $\U^T + a$.
	This notation works for any possibilistic state that is an affine subspace.}%
	\begin{equation}
		\bigl( \S + \Q , s + a \bigr)   \qquad \text{with support} \qquad  \bigl(  \P  \cap \A \bigr) + a + s.
	\end{equation}
	By \cref{prop:maps_preserve_supp} (extended from epistemic states to all possibilistic state that are affine subspaces), the possibilistic state after the measurement interaction $m$ is
	\begin{equation}
		\left( \bigl( M^T \bigr)^{-1}  (\S + \Q) , \, m(s + a) \right), 
	\end{equation}
	which becomes
	\begin{equation}
		\left( \bigl[ \bigl( M^T \bigr)^{-1}  (\S + \Q) \bigr] \cap \A , \, A \circ m(s + a) \right)
	\end{equation}
	after the marginalization to $\A$ via the projection $A \colon \S \oplus \A \to \A$. 
	Thus, the condition that this is another pointer state amounts to 
	\begin{equation}
		\Q = \bigl[ \bigl( M^T \bigr)^{-1}  (\S + \Q) \bigr] \cap \A.
	\end{equation}
	The inclusion of $\Q$ within the right hand side is equivalent to 
	\begin{equation}
		q \in \Q \quad \implies \quad M^T q \in \S \oplus \Q.
	\end{equation}
	The reverse inclusion, on the other hand, is equivalent to
	\begin{equation}
		q \in \A  \text{ and } M^T q \in \S \oplus \Q  \quad \implies \quad  q \in \Q,
	\end{equation}
	so that the result follows.
\end{proof}
\begin{remark}\label{rem:block_measurement}
	Note that if we write $M$ in block form as in \cref{eq:block_measurement}, then condition (ii) says that 
	\begin{itemize}
		\item $M_{\Q \P } \colon  \P  \to \Q$ is equal to $0$, and 
		\item $M_{ \P   \P } \colon  \P  \to  \P $ is non-degenerate.
	\end{itemize}
\end{remark}
\begin{proof}
	First, consider an arbitrary $q \in \Q$, so that in this block form we have
	\begin{equation}
		M^T q = M^T \colvec{0,q,0} = \colvec{M_{\Q\S}^T q, M_{\Q\Q}^T q, M_{\Q  \P } q}.
	\end{equation}
	Thus $M^T q$ is an element of $\S \oplus \Q$ if and only if $M_{\Q  \P }$ vanishes.
	
	On the other hand, consider an arbitrary $a \in \A$ with components $a_\Q$ and $a_{\P }$ in $\Q$ and $ \P $ respectively.
	Then, using $M_{\Q  \P } = 0$, we have 
	\begin{equation}
		M^T a = M^T \colvec{0, a_\Q, a_{ \P }} = \colvec{M_{\Q\S}^T \, a_\Q + M_{ \P  \S}^T \, a_{ \P }, M_{\Q\Q}^T \, a_\Q + M_{ \P  \Q}^T \, a_{ \P }, M_{ \P   \P }^T \, a_{ \P }}, 
	\end{equation}
	so that the implication
	\begin{equation}
		M^T a \in \S \oplus \Q  \quad \implies \quad  a \in \Q
	\end{equation}
	becomes 
	\begin{equation}
		M_{ \P   \P }^T \, a_{\P } = 0  \quad \implies \quad  a_{ \P } = 0.
	\end{equation}
	This implication is satisfied if and only if $M_{\P\P }^T$ is non-degenerate, which is equivalent to $M_{ \P   \P }$ itself being non-degenerate.
\end{proof}

\end{document}